\documentclass[10.5pt,twoside,letter]{article}
\usepackage{amsmath,amssymb,euscript}
\usepackage{amsmath,amsthm,amssymb,mathrsfs}
\usepackage{enumerate}

\newtheorem{thm}{Theorem}
\newtheorem{lem}{Lemma}
\newtheorem{prop}{Proposition}
\newtheorem{rem}{Remark}
\newtheorem{corollary}{Corollary}

\begin{document}

\newcommand{\nc}{\newcommand}

\nc{\be}{\begin{equation}}
\nc{\la}{\label}
\nc{\ba}{\begin{array}}
\nc{\ea}{\end{array}}
\nc{\bs}{\begin{split}}
\nc{\es}{\end{split}}

\nc{\R}{\mathbb R}
\nc{\Z}{\mathbb Z}
\nc{\C}{\mathbb C}
\nc{\J}{\mathbb J}

\nc{\p}{\partial}
\nc{\ra}{\rightarrow}
\nc{\ran}{\rangle}
\nc{\lan}{\langle}

\nc{\zb}{\underbar{z}}
\nc{\pb}{\underbar{p}}
\nc{\bb}{\underbar{b}}
\nc{\zbt}{\underbar{z}(t)}

\nc{\dA}{\nabla_A}

\nc{\nj}{(n)}
\nc{\nk}{(n)}
\nc{\nl}{(n)}
\nc{\nr}{(n)}
\nc{\nt}{(n)}

\nc{\vz}{v}
\nc{\psiz}{\psi}
\nc{\Az}{A}
\nc{\Bz}{B_{\zb,\chi}}
\nc{\jz}{j_{\zb,\chi}}
\nc{\Tz}{T}
\nc{\Gz}{G}
\nc{\Lz}{L_0}
\nc{\bLz}{\bar{L}}
\nc{\Pz}{P}
\nc{\bPz}{\bar{P}}
\nc{\wz}{w_{\chi}}

\nc{\E}{{\cal E}}

\nc{\e}{\epsilon}
\nc{\lam}{\lambda}
\nc{\G}{\Gamma}
\nc{\g}{\gamma}
\nc{\al}{\alpha}
\nc{\del}{\delta}
\nc{\Om}{\Omega}
\nc{\Omt}{\tilde{\Omega}}
\nc{\ta}{\tau}
\nc{\w}{\omega}
\nc{\io}{\iota}
\nc{\h}{\theta}
\nc{\z}{\zeta}
\nc{\ka}{\kappa}
\nc{\s}{\sigma}
\nc{\Si}{\Sigma}
\nc{\Lam}{\Lambda}

\nc{\bP}{\bar{P}}
\nc{\bQ}{\bar{Q}}
\nc{\bL}{\bar{L}}

\nc{\chit}{\tilde{\chi}}
\nc{\It}{\tilde{I}}
\nc{\alt}{\tilde{\alpha}}
\nc{\Mt}{\tilde{M}}
\nc{\gt}{\tilde{\gamma}}
\nc{\Jt}{\tilde{J}}

\nc{\bfone}{{\bf 1}}

\newcommand{\cA}{\mathcal{A}}
\newcommand{\cB}{\mathcal{B}}
\newcommand{\cC}{\mathcal{C}}
\newcommand{\cD}{\mathcal{D}}
\newcommand{\cE}{\mathcal{E}}
\newcommand{\cF}{\mathcal{F}}
\newcommand{\cG}{\mathcal{G}}
\newcommand{\cH}{\mathcal{H}}
\newcommand{\cI}{\mathcal{I}}
\newcommand{\cJ}{\mathcal{J}}
\newcommand{\cK}{\mathcal{K}}
\newcommand{\cL}{\mathcal{L}}         
\newcommand{\cM}{\mathcal{M}}         
\newcommand{\cN}{\mathcal{N}}         
\newcommand{\cO}{\mathcal{O}}         
\newcommand{\cP}{\mathcal{P}}         
\newcommand{\cQ}{\mathcal{Q}}
\newcommand{\cR}{\mathcal{R}}
\newcommand{\cS}{\mathcal{S}}
\newcommand{\cT}{\mathcal{T}}
\newcommand{\cU}{\mathcal{U}}
\newcommand{\cV}{\mathcal{V}}
\newcommand{\cW}{\mathcal{W}}
\newcommand{\cX}{\mathcal{X}}
\newcommand{\cY}{\mathcal{Y}}
\newcommand{\cZ}{\mathcal{Z}}
\newcommand{\Zm}{\mathscr{Z}}

\newcommand{\one}{\mathbf{1}}
\newcommand{\id}{\mathbf{1}}
\newcommand{\1}{{\mathbf{1}}}
\newcommand{\const}{\operatorname{const}}

\newcommand{\LAT}{\mathcal{L}}
\newcommand{\Lat}{\mathcal{L}}
\renewcommand{\Re}{\operatorname{Re}}
\renewcommand{\Im}{\operatorname{Im}}
\newcommand{\IM}{\operatorname{Im}}
\newcommand{\RE}{\operatorname{Re}}

\newcommand{\Null}{\operatorname{Null}}
\newcommand{\Ran}{\operatorname{Ran}}
\newcommand{\NULL}{\operatorname{Null}}
\newcommand{\RANGE}{\operatorname{Ran}}
\newcommand{\supp}{\operatorname{Supp}}
\newcommand{\Span}{\operatorname{Span}}
\newcommand{\sign}{\operatorname{sign}}
\renewcommand{\arg}{\textrm{arg}}

\newcommand{\Curl}{\operatorname{curl}}
\newcommand{\curl}{\operatorname{curl}}
\newcommand{\CURL}{\operatorname{curl}}
\newcommand{\Div}{\operatorname{div}}
\renewcommand{\div}{\operatorname{div}}
\newcommand{\DIV}{\operatorname{div}}
\newcommand{\Cov}[1]{\nabla_{\!\!#1}}

\newcommand{\COVGRAD}[1]{\nabla_{\!\!#1}}
\newcommand{\COVLAP}[1]{\Delta_{\!#1}}

\newcommand{\mHs}{\mathscr{H}^{r}(\Omega)}
\newcommand{\mHr}{\mathscr{H}^{r}(\Omega)}
\newcommand{\mH}{\mathscr{H}}
\newcommand{\Hm}[2]{\mathscr{H}^{2}(\Omega)}
\newcommand{\Lpsi}[2]{\mathscr{L}_{#2}^{}(\Omega)}
\newcommand{\Hpsi}[2]{\mathscr{H}_{#2}^{}(\Omega)}
\newcommand{\LA}[2]{\vec{\mathscr{L}}_{}^{}(\Omega)}
\newcommand{\HA}[2]{\vec{\mathscr{H}}_{}^{}(\Omega)}


\newcommand{\CELLAVG}[2]{\left\langle #2 \right\rangle_{#1}}
\newcommand{\DOT}[3]{\left\langle #2, #3 \right\rangle_{#1}}
\newcommand{\NORM}[2]{\left\| #2 \right\|_{#1}}

\newcommand{\TwoByOne}[2]{\left( \begin{array}{c} #1 \\ #2 \end{array} \right)}
\newcommand{\TwoByTwo}[4]{\left( \begin{array}{cc} #1 & #2 \\ #3 & #4 \end{array} \right)}
\newcommand{\ThreeByThree}[9]{\left( \begin{array}{c} #1 & #2  & #3  \\ #4 & #5  & #6  \\ #7 & #8  & #9  \end{array} \right)}
\newcommand{\FourByOne}[4]{\left( \begin{array}{c} #1 \\ #2 \\ #3 \\ #4 \end{array} \right)}

\newcommand{\donothing}[1]{}
\newcommand{\DETAILS}[1]{}

\newcommand{\DATUM}{July 21, 2011}              
\pagestyle{myheadings}                         


\title{Abrikosov Vortex Lattices at Weak Magnetic Fields}

\author{Israel Michael Sigal
\thanks{The corresponding author, im.sigal@utoronto.ca; Dept.~of Math.,
Univ. of Toronto, Toronto, Canada; Supported by NSERC Grant No. NA7901}\\
\and
Tim Tzaneteas
\thanks{ Dept.~of Math.,
Univ. of Braunschweig, Germany} \\
}

\date{}
\maketitle

\begin{abstract}
  We  prove existence of Abrikosov vortex lattice solutions of the Ginzburg-Landau equations in two dimensions,  for magnetic fields larger than but close to the first critical magnetic field. 

\end{abstract}


\section{Introduction}

\subsection{Ginzburg-Landau equations}

In this paper we prove existence of Abrikosov lattice solutions of Ginzburg-Landau equations of superconductivity at weak magnetic fields. In the  Ginzburg-Landau theory the equilibrium configurations are
described by the Ginzburg-Landau equations:
\begin{equation}\la{gle}
   \ba{c}
   -\Delta_A \Psi - \kappa^2 (1-|\Psi|^2) \Psi =0,  \\
   \curl^2 A - \Im ( \bar{\Psi} \nabla_A \Psi )=0,
   \ea
\end{equation}
where  $\Psi :\R^2 \to \C$ is the order parameter,  $A: \R^2 \to \R^2$ is the vector potential of the magnetic field $B(x) := \curl A(x)$, and 
$\nabla_A = \nabla - iA$, and $\Delta_A = -\nabla^*_A\nabla_A$,
the covariant gradient and covariant Laplacian, respectively.  
$|\Psi(x)|^2$ gives the local density of (Cooper pairs of) superconducting electrons and 
the vector-function
$  J(x) := \Im ( \bar{\Psi}(\nabla - iA)\Psi )$, on the r.h.s. of the second equation, is
the superconducting current.

The parameter $\kappa $ is a material constant depending, among other things, on the temperature.  It is 
called the Ginzburg-Landau parameter and it is the ratio of the length scale for $A$ (penetration depth) to the length scale for $\Psi$ (coherence length). The value $\kappa = 1/\sqrt{2}$ divides all superconductors into two groups,  type I superconductors ($\kappa<1/\sqrt{2}$) and type II superconductors ($\kappa>1/\sqrt{2}$).


The Ginzburg-Landau
equations~\eqref{gle} have the trivial solutions corresponding
to physically homogeneous states:
\begin{enumerate}
\item the perfect superconductor solution:
$(\Psi_s \equiv 1, A_s \equiv 0)$
(so the magnetic field $B_s= \Curl A_s \equiv 0$),
\item
the normal metal solution: $(\Psi_n \equiv  0, A_n)$, the magnetic field $B_n = \Curl A_n$ is constant.
\end{enumerate}
We see that the perfect superconductor is a solution only when the magnetic field $B=\curl A$ 
is zero. On the other hand, there is a normal solution 
for any constant $B$. 


Though the equations \eqref{gle} depend explicitly on only one parameter, $\kappa $, there is another - hidden - parameter determining solutions. It can be alternatively expressed as the average magnetic field, $b$, in the sample, 
or as an applied magnetic field, $h$. As it increases in type II superconductors  from $0$, the pure superconducting state turns into a mixed state, which after further increase becomes the normal state. (For type I superconductors, the behaviour is quite different: the transitions from superconducting to normal state and back are abrupt and occur at different values of magnetic field - hysteresis behaviour.)

One of the greatest achievements of the Ginzburg-Landau theory of superconductivity is the discovery by A.A. Abrikosov (\cite{abr}) of solutions with symmetry of square and triangular lattices (\textit{Abrikosov vortex lattice solutions}) and one unit of magnetic flux per lattice cell, for type II superconductors in the regime just before the mixed state becomes the normal one (the regime \eqref{regime1} below). The rigorous proof of existence of such solutions was provided in \cite{odeh, lash, bgt, dut, Al, ts}. Moreover,  important and fairly detailed results on asymptotic behaviour of solutions, for $\ka \ra\infty$ and the applied magnetic fields, $h$, satisfying $h\le \frac{1}{2}\log \kappa+\const$ (the London limit), were obtained in \cite{as} (see this paper and the book \cite{ss} for references to earlier works). Further extensions to the Ginzburg-Landau equations for anisotropic and high temperature superconductors can be found in \cite{abs1, abs2}.

In this paper we prove existence of Abrikosov lattice solutions in the regime just after the superconducting state became the mixed one (the regime \eqref{regime2} in Appendix \ref{sec:crit-mf}) 
for all values of the Ginzburg-Landau parameter $\kappa$'s, all lattice shapes and all (quantized) values of magnetic flux per lattice cell. We also show that in each lattice cell, the solution looks like and $n$-vortex place at the center of the cell.

\bigskip
\subsection{Ginzburg-Landau free energy}

The Ginzburg-Landau equations are Euler-Lagrange equations for the Ginzburg-Landau (Helmholtz) free energy
 \be \la{glen}
   \E_Q(\Psi,A) := \frac{1}{2} \int_{Q} \left\{
  |\nabla_A \Psi|^2  + \frac{\kappa^2}{2}(1 - |\Psi|^2)^2  + (\curl A)^2\right\}d^2x,
\end{equation}
where $Q$ is the domain occupied by the superconducting sample. This energy depends on the temperature (through $\kappa$) and  the average magnetic field, $b=\lim_{Q'\ra Q}\frac{1}{|Q'|}\int_{Q'} \curl A$, in the sample, as thermodynamic parameters. Alternatively, one can consider the free energy depending on the temperature and  an applied magnetic field, $h$. This leads (through the Legendre transform) to 
the Ginzburg-Landau Gibbs free energy $G_Q(\Psi,A) :=\E_Q(\Psi,A)  -\Phi_Q h,$
where $\Phi_Q=b|Q|=\int_Q \curl A$ is the total magnetic flux through the sample.  
$b$ or $h$ do not enter the equations \eqref{gle} explicitly, but they determine the density of vortices, which we describe below.

\bigskip

%

\subsection{Symmetries and equivariant solutions}
The Ginzburg-Landau equations~\eqref{gle} admit several
symmetries, that is, transformations which map solutions to
solutions:

{\it Gauge symmetry}:  for any sufficiently regular function $\eta : \R^2 \to \R$,
\begin{equation}\label{g-transf}
  \G_\g:\  (\Psi(x) ,\   A(x)) \mapsto ( e^{i\eta(x)}\Psi(x),\ A(x) + \nabla\eta(x));
\end{equation}

{\it Translation symmetry}: for any $h \in \R^2$,
\begin{equation}\label{tr-transf}
 T_h:\   (\Psi(x),\  A(x))  \mapsto (\Psi(x + h),\ A(x + h));
\end{equation}

{\it Rotation and reflection symmetry}: for any $R \in O(2)$ (including  the reflections

$f(x)\ra f(-x)$)
\begin{equation}\label{rot-transf}
  T_R:\    (\Psi(x),\  A(x))  \mapsto (\Psi(Rx),\ R^{-1}A(Rx)).
\end{equation}

The symmetries allow us to introduce special classes of solutions, called equivariant solutions. They are defined as solutions having the property that they are gauge equivalent under the action, $T$, of a subgroup, $G$, of the group of rigid motions which is a semi-direct product of the groups of translations and rotations, i.e., for any $g\in G$, there is $\g=\g(g)$ s.t.
$$T_g (\Psi, A) =\G_\g (\Psi, A), $$ 
where $T_g$ 
for the groups of translations, and  
 rotations, is given \eqref{tr-transf} and \eqref{rot-transf}, respectively, and    $\G_\g$ is the action of for the gauge group, given in \eqref{g-transf}.

For $G$ the group of rotations, $O(2)$, we arrive at the notion of the (magnetic) vortex, which is labeled by the equivalence classes of the homomorphisms of $S^1$ into $U(1)$, i.e. by integers $n$,


\begin{equation} \label{vort}
   \Psi^{(n)} (x) = f^{(n)} (r) e^{in\theta} {\hbox{\quad and \quad}}
   A^{(n)}(x) = a^{(n)} (r) \nabla (n\theta) \ ,
\end{equation}
where $(r,\theta)$ are the polar
coordinates of $x \in \R^2$. Such vortices exist and are unique, up to symmetry transformation, for every $n\in \Z$ and their profiles have the following properties 
(see \cite{gst} and references therein):  
\be \label{vortdecay}
  \ba{c}
  |\p^\alpha (1-f^{(n)} (r))| \leq c e^{-m_{\kappa} r}, \\
  |\p^\alpha (1-a^{(n)} (r))| \leq c e^{- r},
  \ea
\end{equation}
\be \label{vortatzero}
f^{(n)} (r) 
= r^n + O(r^{n+2})\ \quad \mbox{and}\  \quad  a^{(n)} (r) 
=r^2 + O(r^{4}),\ \mbox{as}\ r\ra 0.
\end{equation}
Here  $m_\kappa := \min(\sqrt{2}\kappa,1)$. The exponential decay rates at infinity for   $f^{(n)} (r)$ and $a^{(n)} (r)$ are called the \textit{coherence length} and \textit{penetration depth}, respectively.

For $G$ a finite subgroup of the group of rotations, $O(2)$, say $C_{k}$ (see \cite{dfn}), a possible solution would be a polygon of vortices, similar to the one described in \cite{os}.

If  $G$ is the subgroup of the group of lattice translations for a lattice $\cL$, then we call the corresponding solution a \emph{lattice}, or  $\cL$-\textit{gauge-periodic state}. Explicitly,
\begin{equation} \label{gaugeper}
	T_s (	\Psi(x), A(x)) = \G_{g_s(x)}(\Psi(x),  A(x)),\ \forall s\in \cL,
	\end{equation}
where 
$g_s : \R^2 \to \R$ is, in general, a multi-valued differentiable function, with differences of values at the same point  $\in 2\pi\Z$, and satisfying
\be  \la{gspercond}
  g_{s+t}(x)-g_{s}(x+t)-g_{t}(x)\in 2\pi\Z.
\end{equation}
The latter condition on $g_s$ can be derived by computing $\Psi(x + s+t)$ in two different ways.
In the special case described above this  is the Abrikosov (vortex) lattice.

The characteristic property of  $\cL$-gauge-periodic states is their physical characteristics $|\Psi|^2,\ B(x)$ and $  J(x)$, where, recall $ B(x) := \curl A(x)$ and
$  J(x) := \Im (\bar{\Psi}\nabla_A \Psi)$, are doubly periodic with respect to the lattice $\cL$. The converse is also true: a state whose  physical characteristics  are doubly periodic with respect to some lattice $\cL$ is a $\cL$-gauge-periodic state.
\bigskip

\DETAILS{
\subsection{Lattice states} \la{lattstates}

Given a lattice $\cL\subset \R^2$, a \emph{lattice}, or  $\cL$-\textit{gauge-periodic state}, is a pair $(\Psi, A)$ defined on all of $\R^2$ whose physical characteristics $|\Psi|^2,\ B(x)$ and $  J(x)$, where, recall $ B(x) := \curl A(x)$ and
$  J(x) := \Im (\bar{\Psi}\nabla_A \Psi)$, are doubly periodic with respect to some lattice $\cL$. Due to the gauge symmetry of the equations this means that the translations of the pair under this lattice are mutually gauge equivalent, i.e.  for each $s \in \LAT$, there exists in general a multi-valued differentiable function $g_s : \R^2 \to \R$, with differences of values at the same point  $\in 2\pi\Z$, so that
\begin{equation} \label{gaugeper}
	\begin{cases}
	\Psi(x + s) = e^{ig_s(x)}\Psi(x), \\
	A(x + s) = A(x) + \nabla g_s(x).
	\end{cases}
\end{equation}
 Computing $\Psi(x + s+t)$ in two different ways, we find the following condition on $g_s$
\be  \la{gspercond}
  g_{s+t}(x)-g_{s}(x+t)-g_{t}(x)\in 2\pi\Z.
\end{equation}}

An important property of lattice states is flux quantization: The flux, $\int_\Om \curl A$, through the fundamental lattice cell $\Om$ (and therefore through any lattice cell) is
\be  \la{flux-quant}
  \int_\Om \curl A= 2\pi n,
\end{equation}
for some integer $n$. (Indeed, if $|\Psi| > 0$ on the boundary of the cell, we can write 
$\Psi = |\Psi|e^{i\chi}$, for $0 \leq \chi < 2\pi$. The periodicity of $|\Psi|^2$ and $J(x) := \Im ( \bar{\Psi}(\nabla - iA)\Psi )$ ensure the periodicity of $\nabla\chi - A$ and therefore by Green's theorem, $\int_\Omega \Curl A = \oint_{\partial\Omega} A = \oint_{\partial\Omega} \nabla\chi$ and this function is equal to $2\pi n$ since $\Psi$ is single-valued.) 
Now, due to \eqref{gaugeper}, the equation $\int_{\p\Om} A=2\pi n$ is equivalent to the condition
\be  \la{gsfluxcond}
 -\int_{\p_1\Om} \nabla g_{\w_1}(x)+\int_{\p_2\Om} \nabla g_{\w_2}(x)= 2\pi n,
\end{equation}
where $\{\w_1, \w_2\}$ is the basis of $\Om$ and $\p_1\Om/\p_2\Om$ is the part of the boundary of  $\Om$ parallel to $\w_2/\w_1$. Finally, note that the flux quantization can be written as
	$b = \frac{2\pi n}{|\Omega|},$ 
where $b$ is the average magnetic flux per cell,
	$b = \frac{1}{|\Omega|}\int_\Om  \CURL A .$	
Using the reflection symmetry, we can assume that $b$, and therefore $n$, is positive.

\bigskip

\subsection{ Parity}\label{subsec:parity}
It is 
convenient to restrict the class of solutions we are looking for as follows.  We place the co-ordinate origin at the center of the fundamental cell $\Omega$ so that $\Omega$ (as well as $\cL$) is
\textit{invariant under the reflection} $x \ra -x$. 
We define reflection (parity) operation
\be \la{R}
 Rf(x)=f(-x).
\end{equation}  
We say that a function $f$ on $\R^2$, or on the fundamental cell $\Omega$ is even or odd, if it is even or odd under reflection in any cite of the lattice. A pair $w =(\xi, \alpha)$ of functions on $\R^2$, or on  $\Omega$ is said to be even/odd if its $\psi-$ and $a-$component are even/odd and odd, respectively. 
Note that, since $\theta(-x)=\theta(x)+\pi$, the $n-$vortex solutions,  $U^{(n)}:=(\Psi^{(n)}, A^{(n)})$, are odd, if $n$ are odd, and are even, if $n$ are even.

Since the Ginzburg-Landau equations \eqref{gle}, the fundamental cell $\Omega$ and the lattice $\cL$ are
invariant under the reflection $x \ra -x$, we can restrict ourself to either odd or even lattice state  solutions. For convenience, we consider in what follows only odd solutions and odd vortices: 
\be \la{evensol}
 (\Psi(x), A(x))\ \quad \mbox{and}\ \quad  n \quad \mbox{are \textit{odd}}.
\end{equation}
Even solutions and $n$ are treated in exactly the same way.

\subsection{First result: Existence of vortex lattice states} \la{subsec:res1}

We describe here our main result.  First we identify $\R^2$ with $\C$ and note that any lattice $\mathcal{L} \subseteq \C$ can be given a basis ${r, r'}$ such
that the ratio $\tau = \frac{r'}{r}$ satisfies the inequalities $|\tau| \geq 1$,
 $\Im\tau > 0$,  $-\frac{1}{2} < \Re\tau \leq \frac{1}{2}$, and $\Re\tau \geq 0$ if $|\tau| = 1$  (see \cite{Ahlfors}, where the term discrete module, rather than lattice, is used). Although the basis is not unique, the value of $\tau$ is, and we will use that as a measure of the \textit{shape} of the lattice.
 Let $\cL\equiv\cL_R$ be a family of lattices of a fixed shape, with the minimal distance $R\gg 1$ between the nearest neighbour sites. Then the area of the fundamental cells, $\Om$, of $\cL$ is $\ge R^2$ and the average magnetic field $b=O(R^{-2})$. 
We have 
\begin{thm} \label{thm:existAL}
Let $\kappa\ne \frac{1}{\sqrt{2}}$ and $n \ne 0$. For any $n\in \Z$ 
there is  $R_0=R_0(\kappa )\ (\sim (\kappa- 1/\sqrt{2})^{-1}) > 0$ such that for $R  \ge R_0$, there exists a $\cL-$periodic, odd solution $U^\cL \equiv (\Psi^\cL, A^\cL)$ of~\eqref{gle} on the space $\R^2$, s.t. for any $\alpha \in \cL$ we have on $\Omega+\alpha$ 
\be  \la{eq:close}
  U^\cL (x)= U^{(n)}(x-\alpha)+O(e^{-c R}),
\end{equation}
where, recall, $U^{(n)}:=(\Psi^{(n)}, A^{(n)})$ is the $n-$vortex and $c>0$, in the sense of the local Sobolev norm of any index.
\end{thm}

\bigskip
\textbf{Discussion of the result.} 

1) Theorem \ref{thm:existAL} shows that, for every $\kappa \ne 1/\sqrt{2}$ and every lattice shape $\tau$, there is a unique, up to symmetries, Abrikosov lattice solution, $ (\Psi^\cL, A^\cL)$, of the Ginzburg-Landau equations  \eqref{gle}, satisfying \eqref{eq:close} (and \eqref{gaugeper}), as long as $R$ sufficiently large.
(Existence for $\kappa = 1/\sqrt{2}$ is actually trivial.)

2) One can modify our proof to make $R_0$ uniform in $\kappa - 1/\sqrt{2}$, see Remark \ref{rem:unif-kappa}.

3) Let $U_s:=(\Psi_s=1, A_s=0)$, the pure superconducting state and $h_{c1}:=\frac{E^{(1)}}{\Phi^{(1)}}$, where $E^{(n)}:=E(U^{(n)})$ and $\Phi^{(n)}:=\int B^{(n)}$, the energy and  flux of individual $n-$vortex, respectively, the first critical magnetic field (see Appendix \ref{sec:crit-mf}).  For $R$ sufficiently large and for the applied magnetic field $h> h_{c1}$, we have, for the fundamental cell $\Om$,  that the Gibbs energy satisfies
\[G_\Om(U^\cL)<G_\Om(U_s).\]
Indeed, 
due to \eqref{eq:close}, $G_\Om(U^\cL)=G_\Om(U^{(n)})+O(e^{-c R})$. Hence, since $h_{c1}:=\frac{E^{(1)}}{\Phi^{(1)}}$ and  $G_\Om(U_s)=0$, the result follows. 

4) One expects (based on results of \cite{gs2} on the Ginzburg-Landau energy, that for $\kappa>1/\sqrt{2},\ n=1$ and for $R$ sufficiently large, 
the average energy, $E_\Om(\cL):=\frac{1}{|\Om|}\E_\Om(\Psi^\cL, A^\cL)$, of the fundamental cell $\Om$ of the lattices $\cL$ is minimized by the triangular lattice. 

5) One might be able to prove existence of solutions of the Ginzburg-Landau equation \eqref{gle} in a large domain $Q$, which are close to  the Abrikosov lattice solution $U^\cL:=(\Psi^\cL, A^\cL)$. 
To do this we first construct an almost solution  $\tilde U^\cL:=(\tilde \Psi^\cL, \tilde A^\cL)$ by gluing together $U^\cL$ in $Q'\subset Q$ with an appropriate function in $Q/Q'$. 
This would give us the solution $U^\cL_Q:=(\Psi^\cL_{ Q}, A^\cL_{ Q})$ in $Q$,  close to $U^\cL:=(\Psi^\cL, A^\cL)$. 

\bigskip


\bigskip

Our approach to proving Theorem \ref{thm:existAL} is as follows. First we show that the existence problem on $\R^2$ can be reduced to $\Om$ with the boundary conditions on $\Om$  induced by the periodicity condition \eqref{gaugeper} (Subsection \ref{subsec:reduct-to-Om}).  Then we solve the Ginzburg-Landau equations \eqref{gle} on $\Omega$ with the obtained boundary conditions.  
To this end we construct an approximate solution, $v$ 
 (Subsection \ref{subsec:constraaprsolcell}) and use the Lyapunov-Schmidt reduction to obtain an exact solution 
 (Subsection \ref{subsec:LSdec}).   (Then 
 we glue together copies of the translated and gauged solution on $\Omega$ (according to the prescription of Subsection \ref{subsec:reduct-to-Om}) to obtain a solution on $\R^2$.)

\bigskip
%
\DETAILS{\subsection{ Zero modes}

Before formulating our next result we discuss zero modes of the linearized Ginzburg-Landau equations. These modes are associated with each solution of the Ginzburg-Landau equations \eqref{gle} and are due to 
the translation and gauge symmetries of \eqref{gle}. For instance, for the $n-$vortex $U^{(n)}:=( \Psi^{(n)}, A^{(n)})$, they are
\be \la{ntrmodes}
  T^{(n)}_{k} :=((\nabla_{A^{(n)}})_{k} \Psi^{(n)}(x), \;  B^{(n)}(x) Je_k),
\end{equation}
where $B^{(n)}(x):=\curl A^{(n)}$, and
\be \la{ngmodes}
  \Gz_\g^{(n)} := (i \g \Psi^{(n)}, \nabla\g),
\end{equation}
respectively. 
(See \cite{gs1} for a discussion of $T^{(n)}_{k}$ and $G^{(n)}_{\g}$.) Denote by $F(U)$ the map on the l.h.s. of \eqref{gle} and by 
$L^{(n)} := F'( U^{(n)})$,
its linearization around $U^{(n)}:=( \Psi^{(n)}, A^{(n)})$ ($F'( U)=$ the $L^2-$gradient of $F$ at $U$). Then we have 
\[
  L^{(n)}\Tz^{(n)}_{k} = 0,\ L^{(n)} G^{(n)}_{\g}=0.
\]
We note that for $n$ odd, $T^{(n)}_{k},\ k=1, 2,$ and $G^{(n)}_\g$, $\g$ even,  are even and odd, respectively, since  $U^{(n)}:=( \Psi^{(n)}, A^{(n)})$ is odd.


\DETAILS{Let $L_\e$ be the linearization of the Ginzburg-Landau equations at the solution $ (\Psi_\e, A_\e)$. As was mentioned above, $L_\e$ is a real linear operator, symmetric 
	($\DOT{}{v}{L_\e v'} = \DOT{}{L_\e v}{v'}$) 
with respect to the inner product \eqref{IP}, and Explicitly the translation modes $T_i$, $i = 1,2$, and the gauge modes, $G_\gamma,\ \g:\R^2\ra \R$, are given by
\begin{equation} \label{eq:gauge-zero-modes'}	
	T_i = \TwoByOne{ \left(\COVGRAD{A_\e}\Psi_\e \right)_i }{(\CURL A_\e)Je_i }\ \quad \mbox{and}\ \quad  G_\gamma = \TwoByOne{ i\gamma \Psi_\e }{ \nabla\gamma }.
\end{equation}}
Similarly, for a solution $U^\cL :=(\Psi^\cL, \;  A^\cL)$ 
the generators of translations, 
 \be \la{trmodes}
  T^\cL_{k} :=((\nabla_{A^\cL})_{k} \Psi^\cL(x), \;  B^\cL(x) Je_k),
\end{equation}
where $B^\cL(x):=\curl A^\cL$, and  gauge transformations,
\be \la{gmodes}
  \Gz^\cL_{\g} := (i \g \Psi^\cL, \; \nabla \g), 
\end{equation}
are zero modes of the linearization 
\be \la{L}
  L:=F'(U^\cL). 
\end{equation}
Note that $\Gz_{\g}^\cL, \g\in H_2(\R^2, \R)$, are in the domain of $L$, but $T^\cL_{k}$ are not. Indeed, $B^\cL(x)$ is a $\cL-$periodic function and therefore $\notin  H_2(\R^2, \R^2)$.

Finally, we observe that since the solution $U^\cL \equiv (\Psi^\cL, A^\cL)$ is odd, the linearization operator $L$ commutes with the reflections, \eqref{R}.
\bigskip}
%
%

\subsection{Second result: Spectrum of fluctuations} \la{subsec:spec}

To formulate our second result which concerns the spectrum of fluctuations around the solution $U^\cL\equiv (\Psi^\cL, A^\cL)$ found above, we have to introduce the linearized operators and their zero modes. Denote by $F(U),\ U=(\Psi, A),$ the map defined by the l.h.s. of \eqref{gle}. Let $L_{U^*} := F'( U^*)$ be the linearization of $F(U)$ around a solution $U^*:=( \Psi^*, A^*)$ of \eqref{gle} ($F'( U)=$ the $L^2-$gradient of $F$ at $U$). 
Note that $ L_{U^*}$ 
is a real-linear operator, symmetric, 
	$\DOT{}{v}{L_{U^*} v'} = \DOT{}{L_{U^*} v}{v'}$, 
with respect to the inner product
 \begin{equation} \label{ip}
 \langle w,w'\rangle= \int_{\R^2} (\Re \overline{\xi} \xi'+ \alpha \cdot \alpha'),
\end{equation}
where $w=(\xi,\alpha)$, etc..
Unless $U^*$ is trivial, it breaks the translational and gauge symmetry and as a result the linearized operator $L_{U^*}$ has  translation and  gauge symmetry zero modes:
 $LT^*_{k}=0$, $L\Gz_{\g}^*=0$,  
where $ T^*_{k}(x) :=((\nabla_{A^*})_{k} \Psi^*(x), \;  B^*(x) Je_k)$
and  $\Gz^*_{\g}(x) := (i \g \Psi^*(x), \; \nabla \g(x))$, with
 $B^*(x):=\curl A^*(x)$ and $J$, the symplectic matrix
    \begin{equation*}
        J = \left( \begin{array}{cc} 0 & -1 \\ 1 & 0 \end{array} \right).
    \end{equation*}
In particular, this applies to $U^*=U^\cL,\ U^{(n)}$, with the corresponding zero modes denoted by $T_{k},\ \Gz_\g$ and $T^{(n)}_{k},\ \Gz_\g^{(n)}$, respectively, so that e.g.
\be \la{ntrmodes}
  T^{(n)}_{k}(x) :=((\nabla_{A^{(n)}})_{k} \Psi^{(n)}(x), \;  B^{(n)}(x) Je_k),
\end{equation}
where $B^{(n)}(x):=\curl A^{(n)}(x)$, and
\be \la{ngmodes} \Gz_\g^{(n)}(x) := (i \g(x) \Psi^{(n)}(x), \nabla\g(x)), \end{equation}
are translation and gauge zero modes, respectively,
zero modes for the $n-$vortex $U^{(n)}:=( \Psi^{(n)}, A^{(n)})$:
$  L^{(n)}\Tz^{(n)}_{k} = 0,\ L^{(n)} G^{(n)}_{\g}=0,$
with $L^{(n)} := F'( U^{(n)})$. (See \cite{gs1} for a discussion of $T^{(n)}_{k}$ and $G^{(n)}_{\g}$.)

Define  the shifted translational zero modes $T_{jk}(x)=T^{(n)}_{k}(x-j) $, 
associated with the $n-$vortices located at the sites $j$ and let $ L:=L_{U^\cL} := F'(U^\cL)$. We emphasize that while $T_{k}(x)$  are zero modes of $L$,  $T_{jk}(x)$ are not. We have

\begin{thm} \label{thm:L}
Suppose either $\kappa > 1/\sqrt{2}$ and $n=1$ or $\kappa < 1/\sqrt{2}$ and $n \ne 0$. 
There is $R_0=R_0(\kappa )\ (\sim (\kappa - 1/\sqrt{2})^{-1}) > 0$ such that for $R  \ge R_0$,  we have

1) [approximate zero-modes of $L$] 
$\|L \Tz_{jk}\|_{H^r} \lesssim e^{-cR} $, for any $r$; 

2) [Coercivity away from the translation and gauge modes]
$\lan \eta, L\eta\ran \ge c'\|\eta\|_{H^1}^2$, for any
\[\eta \perp \Span\{T_{jk},\ G_{\g} |\  \forall j\in \cL,\ k=1,2,\ \g\in H^2(\R^2, \R)\},\]
and $c' >0$ independent of $R$. 
\end{thm}

Above and in sequel, the norms and inner products without subindices stand for those in $L^2$, while the Sobolev norms on $\Om$ are distinguished by the symbol $H^r$ in the subindex.

We prove this theorem in Section \ref{sec:proofthmsLK}. In exactly the same way one proves a similar, but stronger, result about  a complex-linear extension, $K$, of the operator $L$ (the latter result implies the former one). The spectrum of fluctuations around $U^\cL$ is the spectrum of $K$.
\medskip

This paper is self-contained. In what follows we write $e^{-R}$ for $e^{-cR}$.
\DETAILS{The $n$-vortex is a critical point of the Ginzburg-Landau
energy~\eqref{eq:en}, and the second variation of the energy
\[
  L^{(n)} := \mbox{ Hess } \cE (\psi^{(n)}, A^{(n)})
\]
is the linearized operator for the Ginzburg-Landau
equations~\eqref{GES} around the $n$-vortex, acting on the space $X
= L^2(\R^2,\C) \oplus L^2(\R^2,\R^2)$. This implies that $L^{(n)}$
is self-adjoint.

The symmetry group of $\cE(\Psi,A)$, which is infinite-dimensional
due to gauge transformations, gives rise to an infinite-dimensional
subspace of $\Null(L^{(n)}) \subset X$, which we denote here by
$Z_{sym}$, described in the theorem above. We say the $n$-vortex is
{\em (linearly) stable} if for some $c > 0$,
\[
  L^{(n)}|_{Z_{sym}^{\perp}} \geq c,
\]
and {\em unstable} if $L^{(n)}$ has a negative eigenvalue. By this
definition, a stable state is a local energy minimizer which is a
{\it strict} minimizer in directions orthogonal to the infinitesimal
symmetry transformations. An unstable state is an energy saddle
point. The basic result on vortex stability is the following:

\begin{theorem}[\cite{GS}] \label{thm:stab}

\begin{enumerate}
\item
For Type I superconductors, all $n$-vortices are stable.
\item
For Type II superconductors,  the $\pm1$-vortices are stable, while
the $n$-vortices with $|n| \geq 2$, are unstable.
\end{enumerate}
\end{theorem}}
\bigskip

\vspace{3mm}   \noindent {\bf Acknowledgements:}

Part of this work was done while the first author was visiting ETH Z\"urich and ESI Vienna. He is grateful to these institutions for hospitality. The first author is grateful to Stephen Gustafson for important suggestions, and to Yuri Ovchinnikov,  for useful discussions. The authors are grateful to the anonymous referee for several useful remarks and suggestions.
 \vspace{3mm}
\bigskip
\section{Proof of Theorem \ref{thm:existAL}} \label{sec:approach}

In this section we prove Theorem \ref{thm:existAL}, modulo technical statements proved in the next section. 



\subsection{Reduction to the basic cell}\label{subsec:reduct-to-Om}

Assume we are given  a multi-valued differentiable function $g_s : \R^2 \to \R$, with differences of values at the same point  $\in 2\pi\Z$ and satisfying \eqref{gspercond}.
An example of such a function is $g_s : =\frac{b}{2}s\wedge x =-\frac{b}{2}s\cdot J x$ used in \cite{ts}. Another example will be given below. Given a continuous function $U \equiv (\Psi, A)$  on the space $\R^2$, satisfying the gauge-periodicity conditions \eqref{gaugeper} (a $\cL-$gauge-periodic function), its restriction, $u \equiv (\psi, a)$, to the fundamental cell $\Omega$ satisfies the boundary conditions induced by \eqref{gaugeper}: 
\begin{equation} \label{gaugeperbc}
	\begin{cases}
	\psi(x + s) = e^{ig_s(x)}\psi(x), \\
	a(x + s) = a(x) + \nabla g_s(x), \\
x\in  \p_1\Omega/\p_2\Omega\ \mbox{and}\  s=\w_1/\w_2.
	\end{cases}
\end{equation}
Here $\p_1\Omega/\p_2\Omega=$ the left/bottom boundary of $\Omega$ and $  \{\w_1, \w_2\}$ is a basis in $\cL$.

In the opposite direction, given a continuous function $u \equiv (\psi, a)$ on the fundamental cell $\Omega$, satisfying the boundary conditions \eqref{gaugeperbc},  we lift it to a $\cL-$periodic function $U \equiv (\Psi, A)$  on the space $\R^2$, satisfying the gauge-periodicity conditions \eqref{gaugeper}, by setting,
for any $\alpha \in \cL$,
\be  \la{lifting}
  \Psi (x)= \psi(x-\alpha)e^{i\Phi_\alpha (x)},\   A (x)= a(x-\alpha)+\nabla\Phi_\alpha (x),\   x\in\Omega+\alpha,
\end{equation}
where $\Phi_\alpha (x)$ is a real, possibly multi-valued, function to be determined. (Of course, we can add to it any $\cL-$periodic function.)  We define
\be  \la{Phial}
  \Phi_\alpha (x):=   g_{\al}(x-\al),\  \mbox{for}\  x\in \Omega+\alpha.
\end{equation}
The periodicity condition \eqref{gaugeper}, applied to the cells $\Omega+\alpha-\w_i$ and $\Omega+\alpha$ and the continuity condition on the common boundary of the cells $\Omega+\alpha-\w_i$ and $\Omega+\alpha$ imply that $\Phi_\alpha (x)$ should satisfy the following two conditions:
\be  \la{Phicond1}
  \Phi_\alpha (x)= \Phi_{\alpha-\w_i} (x-\w_i) +g_{\w_i}(x-\w_i),\ \mbox{mod}\ 2\pi,\ x\in \Omega+\alpha,
\end{equation}
\be  \la{Phicond2}
  \Phi_\alpha (x)= \Phi_{\alpha-\w_i} (x) +g_{\w_i}(x-\alpha),\ \mbox{mod}\ 2\pi,\ x\in \p_i\Omega+\alpha,
\end{equation}
where $i=1, 2,$ and, recall, $\{\w_1, \w_2\}$ is a basis in $\cL$ and $\p_1\Omega/\p_2\Omega$ is the left/bottom boundary of $\Omega$. 

To show that \eqref{Phial} satisfies the conditions \eqref{Phicond1} and \eqref{Phicond2}, we note that, due to \eqref{gspercond}, we have $g_{\al}(x-\al)= g_{\al-\w_i}(x-\al) +g_{\w_i}(x-\w_i),\ \mbox{mod}\ 2\pi,\ x\in \Omega+\alpha,$ and $g_{\al}(x-\al)= g_{\al-\w_i}(x-\al+\w_i) +g_{\w_i}(x-\alpha),\ \mbox{mod}\ 2\pi,\ x\in \p_i\Omega+\alpha$, which are equivalent to \eqref{Phicond1} and \eqref{Phicond2}, with \eqref{Phial}.

Finally, note that
 \begin{itemize}
 \item[(a)] Since  $\Psi, A$ satisfy the gauge-periodicity conditions \eqref{gaugeper}  in the entire space $\R^2$ and are smooth in $\R^2/(\cup_{s\in\cL}\p\Om)$, $\nabla_A\Psi$, $\Delta_A\Psi$ and $\curl^2 A$  are continuous and satisfy  the gauge-periodicity condition \eqref{gaugeper};
\item[(b)]  Since $u \equiv (\psi, a)$ satisfies the Ginzburg-Landau equations \eqref{gle} in $\Om$, then $U \equiv (\Psi, A)$  satisfies \eqref{gle} in $\R^2/(\cup_{t\in\cL}S_t\p\Om)$, where $S_t: x\ra x+t$;
\item[(c)]   Since  $\Psi, A$ satisfy the gauge-periodicity conditions \eqref{gaugeper}   in the entire space $\R^2$, we conclude by the first equation in
 \eqref{gle} that  $\Delta_A\Psi$ is continuous and satisfies the periodicity conditions (in the first equation of) \eqref{gaugeper} in $\R^2$ and therefore, by the Sobolev embedding, theorem so is $\nabla_A\Psi$. Hence, by the second equation in \eqref{gle}, $\curl^2 A$ is continuous and satisfies the periodicity conditions \eqref{gaugeper} in $\R^2$.  Therefore, by iteration of the above argument (i.e. elliptic regularity),  $\Psi, A$ are smooth functions obeying \eqref{gaugeperbc} and  \eqref{gle}.
 \end{itemize}
We summarize the conclusions above as
\begin{lem} \la{lem:period-deriv}
Assume twice differentiable functions $(\psi, a)$ on $\Om$ obey the boundary conditions \eqref{gaugeperbc} and the Ginzburg-Landau equations \eqref{gle}. Then the functions $ (\Psi, A)$ constructed in \eqref{lifting} - \eqref{Phial} are smooth in $\R^2$ and satisfy the periodicity conditions \eqref{gaugeper} and the Ginzburg-Landau equations \eqref{gle}.
\end{lem}
%
\DETAILS{For general $g_s(x)$, we define, for $x\in \Omega$ and $\alpha=m_1\w_1+m_2\w_2\in\cL$,
\be  \la{Phigen}
  \Phi_\alpha (x+\alpha):=  \Phi_0 (x)+ \sum_{l=0}^{m_1-1}g_{\w_1}(x+l\w_1)+\sum_{m=0}^{m_2-1}g_{\w_2}(x+m_1\w_1+m\w_2),
\end{equation}
where $\Phi_0 (x)$ is a $\cL-$periodic function. 

For $\tilde g_s(x)=\frac{b}{2}x\wedge s$ and $\alpha=m_1\w_1+m_2\w_2$ the equations \eqref{Phicond1} and \eqref{Phicond2} can be satisfied by
\be  \la{Phispec}
  \Phi_\alpha (x):= m_1\tilde g_{\w_1}(x-\alpha)+m_2\tilde g_{\w_2}(x-\alpha)+\frac{b}{2}\w_1\wedge \w_2.
\end{equation}
The gauge function $g_s(x)$ we use below is different.}
%


\subsection{Existence of solutions in the basic cell} \label{subsec:constraaprsolcell-2}

In what follows we look for \textit{odd} solutions, $(\psi, a)$, of the Ginzburg-Landau equations~\eqref{gle} in $\Om$. Our goal now is prove the following
\begin{thm} \label{thm:existALOmega}
For any $n\in \Z$ 
there is $R_0 > 0$ such that for $R  \ge R_0$, there exists a smooth, odd solution $u^\cL \equiv (\psi^\cL, a^\cL)$ of~\eqref{gle} on the fundamental lattice cell $\Omega$, satisfying the boundary conditions \eqref{gaugeperbc} and the estimate, in a Sobolev norm of arbitrary index, 
\be  \la{apprOmega}
  u^\cL (x)= U^{(n)}(x)+O(e^{-R}).
\end{equation}
\end{thm}
 To prove this theorem we construct an approximate solution of \eqref{gle} on $\Omega$ and then use a perturbation theory (Lyapunov-Schmidt decomposition), starting with this approximate solution. This is done in Subsections \ref{subsec:constraaprsolcell} - \ref{subsec:LSdec}, modulo technical estimates proven in Section \ref{sec:keyprop}.

Using this result and gluing together copies of the translated and gauged solution on $\Omega$, (see Subsection \ref{subsec:reduct-to-Om} and especially \eqref{lifting} and \eqref{Phial}), we derive Theorem \ref{thm:existAL}. 

\subsection{Construction of an approximate solution} \label{subsec:constraaprsolcell}

In this subsection we  construct test functions,  $ (\psi_{0}, a_{0})$, describing a vortex of  the degree $n$, centered at the center of the fundamental cell $\Omega$.  


Let $\eta$ and  $\bar\eta$ be smooth, nonnegative, spherically symmetric (hence even),
cut-off functions on $\Omega$, such that $\eta =1$  on $|x|\le\frac{1}{3}R$ and $\eta =0$ on $\Omega/\{|x|\le\frac{2}{5}R\}$ and
$$|\p^\alpha \eta (x)| \lesssim R^{-|\alpha|}$$ inside $\Omega$ and  $\eta+\bar\eta = 1$ on $\Omega$. 
Fix an odd integer $n$. We define on $\Omega$ 
\begin{equation} \label{psi0a0}
 \psi_{0}(x) := [ f^{(n)} \eta +\bar\eta](x)e^{in\theta(x)},\ 
  a_{0}(x) := [ A^{(n)} \eta +n\nabla\theta\bar\eta](x).
\end{equation}
%
These functions belong to Sobolev spaces  $H^r_{odd}(\Om):=H^r_{odd}(\Om, \C)\times H^r_{odd}(\Om, \R^2)$ of odd functions, for any $r\ge 0$, and satisfy the boundary conditions \eqref{gaugeperbc} with 
 \begin{equation} \label{gs}
g_s(x):= n\theta(x+s)-n\theta(x)\ \mbox{and}\ x\in \R^2.
\end{equation}
Note that, though the function $g_s(x)$  is multi-valued on $\R^2$, it is well-defined  for  $x\in \p_i\Omega$ and $s=\w_i,\ i=1, 2$. Indeed, $g_s(x)$ can be written as
$$g_s(x)=n\int_0^1 dr\frac{Jx\cdot s}{|x+r s|^2}
=n\frac{Jx\cdot \hat s}{\sqrt{|x|^2 - (x\cdot\hat s)^2}}\int_{\lam_1}^{\lam_2} \frac{dt}{t^2+1},$$
where $\lam_1=\frac{x\cdot\hat s}{\sqrt{|x|^2 - (x\cdot\hat s)^2}},\ \lam_2=\frac{|s|+x\cdot\hat s}{\sqrt{|x|^2 - (x\cdot\hat s)^2}},$ and $\hat x:=x/|x|,$ etc. (Note that, taking for simplicity lattices with equal sides, by our choice, $\Om=\{r_1\w_1 + r_2\w_2\ |\ -\frac{R}{2} \le r_i\le \frac{R}{2}\ \forall i\}$, $$\p_1\Omega:=\{-\frac{1}{2}\w_1 + r\w_2\ |\ -\frac{R}{2} \le r\le \frac{R}{2}\}\ \mbox{and}\ 
\p_2\Omega:=\{r\w_1 - \frac{R}{2}\w_2 \ |\ -\frac{R}{2} \le r\le \frac{R}{2}\},$$ so that $|x|^2 - (x\cdot\hat s)^2$ never vanishes for $s=\w_i,\ x\in \p_i\Omega$.) It can be also verified directly that \eqref{gs} satisfies the conditions \eqref{gspercond} and \eqref{gsfluxcond}:
\be  \la{gspercond'}
  g_{s+t}(x)-g_{s}(x+t)-g_{t}(x)\in 2\pi\Z  \,\,\,
\end{equation}
and
\be  \la{gsfluxcond'}
 -\int_{\p_1\Om} \nabla g_{\ell_1}(x)+\int_{\p_2\Om} \nabla g_{\ell_2}(x)=\int_{\p\Om} \nabla n\theta(x)=  2\pi n.
\end{equation}
%
\DETAILS{Compare the above gauge $g_s(x):= n\theta(x+s)-n\theta(x)\ \mbox{and}\ x\in \p\Omega$ with the gauge $\tilde g_s(x)=\frac{b}{2}s\wedge x$, with $s\wedge x = s_1x_2 - s_2x_1\equiv Js\cdot x$, used in \cite{ts}. The latter satisfies
$$\tilde g_t(x)+\tilde g_s(x+t)- \tilde g_{t+s}(x)=\frac{b}{2}t\wedge s=\pi n.$$}
Finally, by the construction we have
$
 \psi_{0} = \Psi^{(n)} +(1- f^{(n)})e^{in\theta}\bar\eta ,\ 
  a_{0} =  A^{(n)}  +n\nabla\theta (1-a^{(n)})\bar\eta.
$
This, the definition of $\bar\eta$   and the estimates \eqref{vortdecay} imply, for $v:=(\psi_{0}, a_{0})$, $U^{(n)} = (\Psi^{(n)}, A^{(n)})$, 
that
\begin{equation} \label{v-Un-est}
\|v-U^{(n)}\|_{H^r}\lesssim e^{-R}\ \quad \forall r\ge 0.
\end{equation}
 \bigskip

\subsection{Spaces} \la{sec:spaces}
We consider the spaces 
$L^2_{odd}(\Om):=L^2_{odd}(\Om, \C)\times L^2_{odd}(\Om, \R^2)$ of odd square integrable functions on $\Om$, with the real inner product 
\eqref{ip}. Fixing an odd integer $n$, we define 
$\mHr$ to be the Sobolev space of order $r\ge 0$ of \textit{odd}   functions $w=\left( \xi, \alpha  \right) : \Omega \ra \C\times \R^2$, 
satisfying 
the gauge periodic boundary conditions
\begin{equation} \label{gaugeperbc'}
	\begin{cases}
	\xi(x + s) = e^{ig_s(x)}\xi(x), \\
	\alpha(x + s) = \alpha(x),
	\end{cases}
\end{equation}
for  $x\in  \p_1\Omega/\p_2\Omega$ (= the left/bottom boundary of $\Omega$), $s=\w_1/\w_2$ ($  \{\w_1, \w_2\}$, a basis in $\cL$), and $g_s$ given in \eqref{gs}. (Note that $\nabla_{a_0}\xi$ satisfies the boundary conditions on the first line on \eqref{gaugeperbc'}.)
For  $r>\frac{1}{2} $, the restrictions of functions in $\mHr$ 
to the boundary exist as  $\mH^{r-\frac{1}{2}}(\Om)$ functions and therefore \eqref{gaugeperbc'} is well defined; for $0\le r\le\frac{1}{2} $, one can define the corresponding spaces by observing that if $\xi \in \mHr$, then $e^{-in\theta}\xi $ is periodic w.r. to the lattice $\cL$ and the corresponding norms can be defined in terms of its `Fourier' coefficients. (We need $\mHr$ for $r=2$.)
\DETAILS{The space $\mHs $ can be mapped by the matrix multiplication operator
\begin{equation} \label{V}
V=\left[ \begin{array}{cc} e^{-in\theta}  & 0
\\0 & \bfone
\end{array} \right]\ :\mathscr{H}_s(\Omega)\ra \mathscr{H}^{per}_s(\Omega),
\end{equation}
 into the Sobolev space, $\mathscr{H}^{per}^s(\Omega) :=H_{s}^{per}(\Omega, \C)\times H^{s}^{per}(\Omega, \R^2),$ of periodic functions on $\Om$.}

\bigskip

\subsection{Generators of translations and gauge transformations} \la{sec:gener}
An important role in the analysis of vortices is played by the generators of translations and gauge transformations, $T_{k},\ k=1, 2,$ and $\Gz_{\g},\ \g:\Omega\ra \R$, defined as
\be \la{trmodes}
  T_{k}(x) :=((\nabla_{a_0})_{k} \psi_0(x), \;  b_0(x) Je_k),
\end{equation}
where $b_0(x):=\curl a_0$, and 
\be \la{gmodes}
  \Gz_{\g} := (i \g \psi_0, \; \nabla \g),\ \gamma : \Omega \ra \R.
\end{equation}
These generators are almost zero modes of the operator  $\Lz := F'(\vz)$ ($=$ the $L^2-$gradient of $F$ at $v$), where, recall, $F$ is the map defined by the l.h.s. of~\eqref{gle}.
\DETAILS{\begin{lem} The operator  $\Lz := F'(\vz)$ is real-linear and symmetric on $L^2(\Omega, \C)\times L^2(\Omega, \R^2)$ and has $T_{k},\ k=1, 2,$ and $\Gz_\g$ as the zero modes: $\Lz T_{k}=0,\ k=1, 2,\ \Lz\Gz_\g=0$.
\end{lem}
\begin{proof} We prove only the last part of this lemma. Denote  $u_\chi = (e^{i\chi}\psi, a+\nabla\chi)$, for any $\chi\in H_2^{per}(\Omega, \R)$. ($u_\chi$ does not satisfy the gauge-periodicity conditions \eqref{gaugeper}, but this does not effect the argument below.) \eqref{eqF} and the gauge symmetry of this equation imply that $ F(u_\chi)=0$. 
Differentiating this equation w.r.to $\chi$ yields the equation $\Lz G_\g=0$.
\end{proof}}

 Since $(\Tz_{k})_\psi$ and $(\Tz_{k})_a$ are even, by our definition in Subsection \ref{subsec:parity}, 
so are $T_{k},\ k=1, 2,$ and therefore $T_{k},\ k=1, 2,$ do not belong to our spaces. On the other hand, $\Gz_{\g}$ belongs to our space $\mathscr{H}^{r}(\Omega),\ \forall r$, iff $\gamma$ is periodic and even, with appropriate smoothness conditions. 


\bigskip

\subsection{Orthogonal decomposition} \la{subsec:orthdec}


Let $v=(\psi_0, a_0)$ with $\psi_0$ and $a_0$ defined in \eqref{psi0a0}. Consider odd functions, $u = (\psi, a)\in L_{odd}^2(\Om)$, 
 satisfying the boundary conditions \eqref{gaugeperbc} with \eqref{gs} and s.t. 
\be \la{orthdecomp}
  u = v + w,\ \mbox{with}\ 
  w\perp \Gz_\g,\  \forall\g\in H_{per}^{2+r}(\Om, \R), 
\end{equation}
where $H_{per}^{2+r}(\Om, \R)$ is the Sobolev space of real, periodic, even functions on $\Om$ of order $2+r$.
%
%
\DETAILS{Here $H_2^{per}(\Omega, \R)$ is the Sobolev space of real, even, periodic functions on $\Omega$.  If we write $w=(\xi, \alpha) 
$, then the gauge orthogonality condition in \eqref{orthdecomp},
\[
      \left\langle \left( \begin{array}{c}
      i\gamma \Psi^{(n)} \\ \nabla \gamma
      \end{array} \right), \left( \begin{array}{c}
      \xi \\ \alpha \end{array} \right) \right\rangle = 0,
\]
for all $\gamma\in H_{per}^{2+r}(\Omega, \R)$, implies, after integration by parts, that on $\Omega$ 
\begin{equation} \label{gchoice}
  \Im(\overline{\Psi^{(n)}}\xi) = \nabla \cdot \alpha. 
\end{equation}
We 
assume from now on that $w=(\xi, \alpha)$ satisfy this equation.}
The function $w$, defined by \eqref{orthdecomp}, has the following properties
\begin{itemize}
\item Since $v= (\psi_0, a_0)$ is odd and since scalar products of even functions with odd ones vanish, $w\perp \Tz_{k},\  k=1,2.$
    \item Since $v$ and $u$ satisfy  the boundary conditions \eqref{gaugeperbc} with \eqref{gs}, we conclude that $w$ satisfies  the boundary conditions \eqref{gaugeperbc'} with \eqref{gs}. 
 \item Since $v\in H^r(\Om)$, for any $r\ge 0$, we have that, if $u\in H^r(\Om)$, then $w\in H^r(\Om)$.
  \end{itemize}
  Note that by integration by parts,  $ w\perp \Gz_\g,\ \forall\g\in H_{per}^{2+r}(\Om, \R),$ is equivalent to
  \be \la{gaugeorthcond'}\Im(\bar\psi_0 \xi)+\div\al=0.
  \end{equation}
%

\subsection{Lyapunov-Schmidt decomposition. 
} \la{subsec:LSdec}
Recall that $F$ is the map defined by the l.h.s. of~\eqref{gle} and denote $u = (\psi, a) : \Omega \ra \C\times \R^2$. The Ginzburg-Landau equations~\eqref{gle} on $\Om$
can be written as
\be \la{eqF}
 F(u)=0.
\end{equation}
Clearly, $F$ maps $v+\mH^{r+2}(\Om)$ to $\mHs$. Let $\Lz := F'(\vz)$. It is a real-linear operator on $L^2(\Omega, \C)\times L^2(\Omega, \R^2)$ 
 mapping $\mH^{r+2}(\Om)$ to $\mHs$. Now, we assume $u\in \mHs,\ r\ge 0$, and 
substitute the decomposition \eqref{orthdecomp} into~(\ref{eqF}) to obtain
\be \la{eqcan}
   F(\vz) + \Lz w + N_{\vz}(w)=0,
\end{equation}
where $N_{\vz}(w)$ is the nonlinearity $N_{\vz}(w):=F(u)-F(\vz)-F'(\vz) w$.

Let $\Pz$ denote the orthogonal projection
from $L^2_{odd}(\Omega, \C)\times L_{odd}^2(\Omega, \R^2)$ onto the subspace 
$\{G_\g,\ |\ \g\in H_{per}^{2+r}(\Om, \R)\}$, and let $\bPz:=\id-\Pz.$ We split \eqref{eqcan} into two equations:
\be \la{eqP}
 \Pz[F(\vz) +\Lz w + N_{\vz}(w)]=0,
\end{equation}
and
\be \la{eqbarP}
 \bPz[F(\vz) +\Lz w + N_{\vz}(w)]=0.
\end{equation}
Recall the notation $\|\cdot\|_{H^r}$ for the norm in the Sobolev space $H^r(\Om)$. The following estimates are proven in Section \ref{sec:keyprop}, $\forall r>r'+1,\ r'\ge 0$,
\be \la{est:F}
\|F(\vz)\|_{H^{r'}}\lesssim e^{-R},
\end{equation}
 \be \la{Linv}
\bLz:=\bPz\Lz\bPz|_{\Ran\bPz}\ \mbox{is invertible and}\   \| \bLz^{-1} \|_{H^{r'}\ra H^{r'+2}} \lesssim 1,
\end{equation}\be \la{est:N}
   \| N_{\vz}(w) \|_{H^{r'}}
  \leq c_r(\|w\|_{H^{r}}^2 + \|w\|_{H^{r}}^3),
\end{equation}
\be \la{est:N-Lip}
  \| N_{\vz}(w')- N_{\vz}(w)  \|_{H^{r'}}\\
 \leq c_r(\|w\|_{H^{r}} + \|w\|_{H^{r}}^2+\|w'\|_{H^{r}} + \|w'\|_{H^{r}}^2)
  \|w' -w\|_{H^{r}}.
\end{equation}
($r=2,\ r'=0$ suffices for us.) 
\begin{prop} \la{prop:w}
Let $n$ be odd and assume  \eqref{est:F} - \eqref{est:N} hold. Then, for $R$ sufficiently large, Eqn \eqref{eqbarP} has a solution, $w=w(v)$, unique in a ball in $H^r$ of the radius $\ll 1$, which is  odd and satisfies  the estimate
\be \la{west}
  \| w \|_{H^r} \lesssim  e^{-R},\ r\ge 1.
\end{equation}
\end{prop}
\begin{proof}
Since the operator $\bLz:=\bPz\Lz\bPz|_{\Ran\bPz}$ is invertible by \eqref{Linv}, the equation \eqref{eqbarP} can be rewritten as
\be \la{eqw}
  w= -\bLz^{-1}\bPz[F(\vz) + N_{\vz}(w)].
\end{equation}
Using the estimates on $F(\vz),\ \bLz^{-1}$ and $N_{\vz}(w)$, given in \eqref{est:F} -- \eqref{est:N-Lip}, one can easily see that the map on the r.h.s. of \eqref{eqw} maps  a ball in $H^r$ of the radius $\ll 1$ into itself and is a contraction, provided 
$R$ is sufficiently large.  Hence the Banach fixed point theorem yields the existence of a unique $w=w( v)$ and the estimate
$$\|w\|_{H^{r'+2}} \lesssim \|F(\vz) \|_{H^{r'}}.$$
This equation together with \eqref{est:F} implies \eqref{west}.
Since $v$ is odd and since  $\bLz^{-1}$ and $N_{\vz}(\cdot)$ are  invariant under the reflections, $w=w( v)$ is odd, by the construction.
\end{proof}

Now we turn to the equation \eqref{eqP}.  
With $u:=v+w( v)$, this equation can be rewritten as
\be \la{gaugeorth}
\langle \Gz_\g, F(u)\rangle=0,\ \forall\g\in H_{per}^{2+r}(\Om, \R). 
 \end{equation}
(Note that Eqn \eqref{eqcan}, the symmetry of $L_0$ and the fact that $\Gz_\g$ is a zero mode of  $L_0$ imply $\langle \Gz_\g, F(u)\rangle=\langle \Gz_\g, N_v(w)\ran$.) 
To show that \eqref{gaugeorth} is satisfied by  $u:=v+w( v)$ we differentiate the equation $\cE_\lam(e^{s\g}\psi, a+s\nabla\g)=\cE_\lam(\psi, a) $, w. r. to $s$ at $s=0$, to obtain
$$ \p_\psi \cE_\lam(\psi, a)i\g\psi+\p_a \cE_\lam(\psi, a)\nabla\g=0,$$ or $ \lan F(\psi, a), G_\g\ran=0$. 
By either varying the Sobolev index $r$ or invoking elliptic regularity one shows smoothness of solutions. This proves, Theorem \ref{thm:existALOmega}, 
 modulo the statements \eqref{est:F} - \eqref{est:N}.
Combining the latter with the lifting procedure,  \eqref{lifting} and \eqref{Phial}, gives Theorem \ref{thm:existAL}. $\Box$

\bigskip

\section{Complex-linear extension $K$ of $L$}\label{subsec:K}

In order to be able to use spectral theory, we construct
a complex-linear extension $K$  of the operator $L$ defined on $\mathcal{H}(\R^2):=L_{odd}^2({\R}^2;\C)\oplus L_{odd}^2(\R^2;\R^2)$, or on $\mathcal{H}(\Om):=L_{odd}^2(\Om;\C)\oplus L_{odd}^2(\Om;\R^2)$, with the boundary conditions \eqref{gaugeperbc'}. 
The (complex-) linear operator $K$ is defined on $\cH^c(\R^2):=[L_{odd}^2({\R}^2;{\C})]^4\equiv [L_{odd}^2({\R}^2;\C)]^2\oplus[L_{odd}^2(\R^2;\C)]^2$, or on $\cH^c(\Om):=[L_{odd}^2({\Om};{\C})]^4$, as follows.  
We first identify $\alpha : \R^2/\Om \to \R^2$ with the function $\alpha^c = \alpha_1 - i\alpha_2 : \R^2/\Om \to \C$.
The space $\cH(\R^2)/\cH(\Om)$ is embedded in $\cH^c(\R^2)/\cH^c(\Om)$ via the isometric injection
\begin{equation}\label{embed}
\s:\	w=\TwoByOne{\xi}{\alpha} \to w^c=\frac{1}{\sqrt{2}}\FourByOne{\xi}{\bar{\xi}}{\alpha^c}{\bar{\alpha^c}},
\end{equation}
in the sense of the inner product \eqref{ip}. (\textit{Below we drop the superscript $^c$ from $\al^c$}.) This embedding transfers the operator $L$ to the operator $L^c\s w:=\s Lw$ 
on the real-linear subspace, $\s\cH$, of  $\cH^c$. (Here $\cH$ ($\cH^c$) stands for either  $\cH(\R^2)$ ($\cH^c(\R^2)$) or $\cH(\Om) $ ($\cH^c(\Om)$).) Next, we define the projection $\pi$ from $\cH^c$ to $\s\cH\subset \cH^c$, by 
\begin{equation}\label{wdeco}
\pi \FourByOne{\xi}{\chi}{\alpha}{\beta}=
		\frac{1}{2} \FourByOne{\xi + \bar{\chi}}{\bar{\xi} + \chi}
				{\alpha + \bar{\beta}}{\bar{\alpha} + \beta},
\end{equation}
and observe that $\pi+i\pi i^{-1} =\1.$  We extend the operator $L^c$ from the subspace $\s\cH$ to the complex-linear operator $K$ on the entire $\cH^c$ as
$$K \hat w:=\s L\pi \hat w+i\s L\pi i^{-1} \hat w.$$ 
\DETAILS{$K\hat w:=(Lw_1)^c+i(Lw_2)^c$, where $\hat w=w_1+iw_2,$
\begin{equation}\label{wdeco}
	\hat w=\FourByOne{\psi}{\phi}{\alpha}{\omega},\ w_1=
		\frac{1}{2} \FourByOne{\psi + \bar{\phi}}{\bar{\psi} + \phi}
				{\alpha + \bar{\omega}}{\bar{\alpha} + \omega},\
		 w_2= \frac{1}{2} \FourByOne{-i(\psi - \bar{\phi})}{i(\bar{\psi} - \phi)}
				{-i(\alpha - \bar{\omega})}{i(\bar{\alpha} - \beta)}.
\end{equation}}
Similarly we proceed on the fundametal cell $\Om$. The explicit form of $K$ is the same on $\cH^c(\R^2)$ and on $\cH^c(\Om)$ and is given in Appendix \ref{app:LK}.

The (complex-) linear operator $K$ has the following properties: 

1) $K=K^*$ on $\cH^c$;


2) $\lan \s w', K\s w\ran  = \lan w', Lw\ran$;

3) $K\s w=\s Lw$; 

4) $0\in \sigma_{ess}(K)$; 

5) $[K, R]=0$ (recall, $R$ is the parity transformation).

Note that the third statement and the property that  $\Gz_{\g}, \g\in H^2(\R^2, \R)$, and  $T_{k},\ k=1, 2$, given above, are a zero mode of $L$ implies that their complexifications, $\Gz_{\g}^c, \g\in H^2(\R^2, \R)$ and  $T_{k}^c,\ k=1, 2$, 
are zero modes of $K:\ K\Gz_{\g}^c =0,\ LT_{k}^c=0$. 
Since 
$T_{k}^c\notin [L^2({\R}^2;{\C})]^4,\ k=1, 2$, but are bounded (or since $ \g\in H^2(\R^2, \R)$),
we have that the operator $K$ defined in the entire $L^2$ space has $0$ in its essential spectrum. 
The same statements, but with 4 replaced by $0\in \sigma (K)$, hold if we replace $\R^2$ by $\Om$.

%
%

%
\DETAILS{\begin{thm} \label{thm:Kspec}
Suppose $\kappa \ne 1/\sqrt{2}$ and $n=1$ if $\kappa > 1/2$. 
There is $R_0 > 0$ such that for $R  \ge R_0$,  
The operator $K$ has a band of the spectrum of order $O(e^{-R})$ around $0$, coming from translational zero modes of the individual vortices, with the rest of the spectrum strictly positive and separated from $0$ by a uniform in $R$ distance. 
\end{thm}}
%
Due to the properties above,  Theorem \ref{thm:L} follows from the following result.
\begin{thm} \label{thm:K}
Suppose $\kappa \ne 1/\sqrt{2}$ and $n=1$ if $\kappa > 1/2$. 
There is $R_0 > 0$ such that for $R  \ge R_0$,  we have

1) [approximate zero-modes] 
$\|K \Tz_{jk}^c\|_{H^s} \lesssim e^{-R};$ 

2) [Coercivity away from the translation and gauge modes]
$\lan \eta, K\eta\ran \ge c\|\eta\|_{H^1}^2$, for any $\eta \perp T_{jk}^c,\ k=1,2,\  \forall j\in \cL,\ G_{\g }^c,\ \g\in H_2(\R^2, \R)$,
and $c >0$ independent of $R$. 
\end{thm}
A proof of Theorem \ref{thm:K}  is identical to the proof of Theorem  \ref{thm:L} and and is given in Section \ref{sec:proofthmsLK}.

Next, we introduce, for an odd integer $n$,  the Sobolev space
$\mH_r^c(\Om)$ of order $r$ of \textit{odd}   functions $w=\left( \xi, \chi, \alpha, \beta  \right) : \Omega \ra \C^4$, 
satisfying the gauge periodic boundary conditions
\begin{equation} \label{gaugeperbc''}
	\begin{cases}
	\xi(x + s) = e^{ig_s(x)}\xi(x), \\
\chi(x + s) = e^{-ig_s(x)}\chi(x), \\
	\alpha(x + s) = \alpha(x), \\
	\beta(x + s) = \beta(x),
	\end{cases}
\end{equation}
for  $x\in  \p_1\Omega/\p_2\Omega$, $s=\w_1/\w_2$, and for $g_s$ given in \eqref{gs}. These conditions extend  \eqref{gaugeperbc'}.

Finally, let $K_0$ be the complex-linear extension of $L_0$, defined as above, i.e. $K_0$ is the restriction of $K$ to $\Om$.   We remark that
$K_0$ maps $\mH_{r+2}^c(\Om)$ into $\mH_r^c(\Om)$, for $s\ge 0$. Moreover, it is shown  in Appendix \ref{app:LK} that
\begin{equation} \label{K0selfadj}
K_0,\ \mbox{defined on}\ L_{odd}^{2}(\Om, \C^4)\ \mbox{with the domain}\ \mH_{2}^c(\Om),\ \mbox{is self-adjoint}.
\end{equation}


 \bigskip


\section{Key properties} \la{sec:keyprop}

In this section we prove the inequalities, 
\eqref{est:F} - 
\eqref{est:N-Lip}, 
used in the proof of Theorem~\ref{thm:existAL}.

\smallskip

\subsection{Approximate static solution property}
\la{sec:approx}

\begin{lem} \la{lem:apprsol}
For $R\ge 1$ and for any $r>0$, we have
\be \la{F-bnd} 
  \| F(\vz) \|_{H^r} \lesssim e^{-\min(\sqrt{2}\kappa, 1)R}.
\end{equation}
\end{lem}
\begin{proof}
The proof is a computation using the facts that $U^{\nj} = (\Psi^{\nj},A^{\nj})$
satisfies the Ginzburg-Landau equations, together with the exponential decay~(\ref{vortdecay}). We write
\begin{equation} \label{psiAcell'}
  \psi_{0} =  \Psi^{(n)} +\xi,\
  a_{0} =    A^{(n)} +\alpha,
\end{equation}
where $\xi$ and $\alpha$ are defined by this expressions.
%
\DETAILS{where $\xi:=e^{in\theta}(1-f^{(n)})\bar\eta$ and $\alpha:= n\nabla\theta(1-a^{(n)})\bar\eta$ 
(see \eqref{vort} and \eqref{psi0a0}). 
The definition of $\bar\eta$  and the exponential decay~(\ref{vortdecay}) give
\begin{equation} \label{al-exp-est}
\xi,\ \al=O(e^{-R})\ \qquad \mbox{ in any Sobolev norm.}
\end{equation}}
Using the first Ginzburg-Landau equation, we find
\be \label{Fpsibound}
  [F(\vz)]_\psi =\Delta_{A^{(n)}}\xi + (2i\alpha\cdot\nabla_{A^{(n)}} +i\div\alpha+|\alpha|^2) (\Psi^{(n)}+\xi)$$
  $$-\kappa^2[(2\Re(\bar  \Psi^{(n)}\xi)+|\xi|^2)( \Psi^{(n)}+\xi)-(1-| \Psi^{(n)}|^2)\xi]. 
\end{equation}
Furthermore, using the second Ginzburg-Landau equation,
$\curl^2 A^{(n)} - \Im ( \bar{\Psi}^{(n)} \nabla_{A^{(n)}} \Psi^{(n)} ) = 0$,  we arrive at
\be \label{FAbound}
  [F(\vz)]_a = \curl^2 \alpha +|\Psi^{(n)}+\xi|^2\alpha$$
  $$-\Im ( \bar{\Psi}^{(n)} \nabla_{A^{(n)}} \xi + \bar{\xi} \nabla_{A^{(n)}} \Psi^{(n)} + \bar{\xi} \nabla_{A^{(n)}} \xi). 
\end{equation}
\textbf{Since by} \eqref{v-Un-est}, $\xi,\ \al=O(e^{-R})$ in any Sobolev norm, the estimates
(\ref{Fpsibound})  and (\ref{FAbound}) imply \eqref{F-bnd}. 
\end{proof}

\smallskip
\subsection{Approximate zero-mode property}\la{sec:zero}


Recall the translational and gauge zero-modes $\Tz_{k},\ k=1, 2,$ and $G_\g$ are given
in~\eqref{trmodes} and \eqref{orthdecomp}.
\begin{lem}[approximate zero-modes] \la{lem:zeromodes}
 For any $k=1, 2,$ $\g$ twice differentiable and bounded together with its derivatives, and $r>0$, we have
\be \la{eq:approx1}
\|\Lz \Tz_{k}\|_{H^r} \lesssim e^{-R},\ \|\Lz G_\g\|_{H^r} \lesssim e^{-R}. 
\end{equation}
\end{lem}
\begin{proof}
Let $L^{(n)} := \E_{GL}''( U^{(n)})$. We may write
\[
  \Lz = L^{(n)} + V^{(n)},
\]
where
 $V^{(n)}$ is
a multiplication operator defined by this relation. Using the explicit form \eqref{Lexpl} of $L$, given in Appendix \ref{app:LK}, we see that $V^{(n)}$ satisfies
\be \la{Vnbnd}
  |V^{(n)}(x)| \lesssim e^{-R}.
\end{equation}
\DETAILS{Let $T^{(n)}_{k},\ k=1, 2,$ and $G_\g^{(n)}$, 
be the vortex generators of translations and gauge transformations, defined as
\be \la{ntrmodes}
  T^{(n)}_{k} :=((\nabla_{A^{(n)}})_{k} \Psi^{(n)}(x), \;  B^{(n)}(x) Je_k),
\end{equation}
where $B^{(n)}(x):=\curl A^{(n)}$, and
\be \la{ngmodes}
  \Gz_\g^{(n)} := (i  \g\Psi^{(n)}, \nabla\g).
\end{equation}}
%
%
Since $\Tz_{k}-T^{(n)}_{k}$ is expressed in terms of $v-U^{(n)}$, the definitions \eqref{ntrmodes}, \eqref{ngmodes}, \eqref{trmodes} and \eqref{gmodes} and the estimates \eqref{v-Un-est}, imply
\be \la{approxzeromodes}\|\Tz_{k}-T^{(n)}_{k}\|_{H^r} \lesssim  e^{-R},\ \|G_\g-G_\g^{(n)}\|_{H^r} \lesssim  e^{-R},
\end{equation}
Using Eqn \eqref{Vnbnd} and \eqref{approxzeromodes}, the definitions of $T^{(n)}_{k}$ and $G_\g^{(n)}$, 
    in \eqref{ntrmodes} and \eqref{ngmodes} and the facts
\[
  L^{(n)}\Tz^{(n)}_{k} = 0,\ L^{(n)} G_\g^{(n)}=0,
\]
we obtain the 
estimates in \eqref{eq:approx1}.
\end{proof}
Recall from Section \ref{subsec:K} that $K_0$ is a complex-linear extension of $L_0$ and the vectors ${T}_{k}^c,\ k=1,2,\ G^c_\g,\ \g\in H_{per}^2(\Om),$ are complexifications of the vectors $ {T}_{k},\ k=1,2,\ G_\g,\ \g\in H_{per}^2(\Om)$, defined in ~\eqref{trmodes} and \eqref{orthdecomp} (see \eqref{embed}).   The properties $\s L_0^{-1}= K_0^{-1}\s$ and $\|\s w\|=\| w\|$ (see Section \ref{subsec:K}) imply
\begin{corollary}[approximate zero-modes] \la{cor:zeromodesK}
We have
\be \la{eq:approx1-K}
\|K_0 \Tz^c_{k}\|_{H^s} \lesssim e^{-R},\ \|K_0 G_\g^c\|_{H^s} \lesssim e^{-R}. 
\end{equation}
\end{corollary}

\bigskip
\subsection{Coercivity of the Hessian}
\la{sec:coer}

In this subsection we prove \eqref{Linv}. 
With the notation as at the end of the last subsection, 
  let $P^c$ be the projection on the span of the vector $ G^c_\g,\ \g\in H_{per}^{2+r}(\Om)$.  We begin with a lower bound on the complexification $K_0$ of $L_0$.
\begin{lem}[coercivity] \la{lem:coer}
 For $R$ sufficiently large and for any $w \in \Ran (\1-P^c)$ and $r\ge 0$, we have
\begin{equation}\label{Klowbnd}
 \|K_0 w \|_{H^{r}}
  \geq c \| w \|_{H^{2+r}}.
\end{equation}
(For $n=1$ if $\kappa>\frac{1}{\sqrt{2}}$ and for any $n$ if $\kappa<\frac{1}{\sqrt{2}}$ we have  $c \|w\|_{H^1}^2 \leq \lan w, K_0 w \ran
  \leq \frac{1}{c} \| w \|_{H^1}^2$, which could be also extended to a larger class of Sobolev spaces.)
\end{lem}
\begin{proof}  We omit the subindex $0$ in $K_0$ and superindex $c$ in $P^c$
and to simplify the exposition we conduct the proof only for $r=0$. The proof for general $r\ge 0$ requires an extra technical step (commuting $(-\Delta+\1)^{\frac{s}{2}}$ through $K_0$).
Let $\{\chi_0,\  \chi_1 \}$ be a partition of unity associated to the
ball of the radius $R/2$ and its exterior, i.e. $\sum_{j=0}^{ 1} \chi_j^2 = 1$,
$\chi_0$ is supported in the ball of  the radius $3R/5$  and $\chi_1$
is supported outside the ball of  the radius $R/2$. We also assume $|\p^\al\chi_j|\lesssim R^{-|\al|}$. Using these properties and commuting $\chi_j$ through $K$, with the help of \[[\chi_j, \Delta]=-2(\nabla\chi_j)\cdot \nabla - (\Delta\chi_j),\]
 we obtain
\[
  \| K  w\|^2 = \sum_0^1 \|\chi_j K  w\|^2 \ge \sum_0^1 \| K \chi_j w\|^2
   -CR^{-2} \|w \|_{H^1}^2.
\]

We extend the function $w$ to an $L^2-$function on $\R^2$ for which we keep the same notation. 
Since $\chi_1$ is supported  outside the ball of  the radius $R/2$,
it follows from Lemma \ref{lem:lb-out} of Appendix \ref{app:LK} that
\[
  \| K \chi_1 w\| \geq   c_1 \|\chi_1 w\|_{H^2},
\]
for some $c_1 > 0$.

Now, since $w \in \Ran (1-\Pz)$, we have that $w \perp  G_\g,\ \g\in H_{per}^2(\Om),$ 
and, since $w$ is odd and  $ {T}^{(n)}_{k},\ k=1,2$, are even, we have that $w \perp {T}^{(n)}_{k},\ k=1,2$.  Therefore, due to \eqref{approxzeromodes}, we have, for the vortex translational and gauge
zero-eigenfunctions,  $ {T}^{(n)}_{k},\ k=1,2$,  $ G^{(n)}_\g,\ \g\in H_{per}^2(\Om),$ 
of $L^{\nj}$,
\[
  |\lan{T}^{(n)}_{k}, \chi_0 w \ran| \lesssim e^{-R},\  |\lan G^{(n)}_\g, \chi_0 w \ran| \lesssim e^{-R},\ \g\in H_{per}^2(\Om).
\]
Let $P^{(n)}$ be the orthogonal projection on the span of $ G^{(n)}_\g,\ \g\in H_{per}^2(\Om),$ and  $ {T}^{(n)}_{k},\ k=1,2$. Writing $K^{\nj} \chi_0 w=K^{\nj} (\1-P^{(n)})\chi_0 w + K^{\nj} P^{(n)}\chi_0 w$ and using the estimate above and 
the $n$-vortex stability result of~\cite{gs1} (see  Theorem \ref{thm:decomp} of Appendix \ref{app:LK}), 
 we obtain
 \[\| K^{\nj} \chi_0 w\|
  \geq c_2 \|\chi_0 w \| - Ce^{-R} \|w\|.\]
Since on the other hand we have trivially that $\| K^{\nj} \chi_0 w\|
  \geq c_3 \|\chi_0 w \|_{H^2} - c_4 \|w\|$, for some $c_3, c_4>0$, the above estimate can be lifted 
  to
\[
  \| K^{\nj} \chi_0 w\|
  \geq c_0 \|\chi_0 w \|_{H^2} - Ce^{-R} \|w\|,
\]
where $K^{\nj}$ is the complex linear extension of $L^{\nj}$. Now, as with $L_0$ in Subsection~\ref{sec:zero}, we write
$K = K^{\nj} + V^{\nj}$, where recall $V^{\nj}$ satisfies the estimate
\[
  |V^{\nj}(x)| \lesssim   e^{-R}.
\]
Then the last two estimates imply
\[
  \| K \chi_0 w\|
  \geq c_0 \|\chi_0 w \|_{H^2} - Ce^{-R} \|w\|.
\]
Collecting the estimates above and using the fact that $\sum\|\chi_j  w\|^2_{H^2} \ge \|  w\|^2_{H^2}- C R^{-2} \|w\|_{H^1}^2 $, we find
  \be
\la{eq:coer}
  \| K w\|^2 \ge  (\min c_j)\|  w\|^2_{H^2}
     - C( e^{-R} + R^{-2}) \|w\|_{H^1}^2,
\end{equation}
which  for $R$ sufficiently large gives \eqref{Klowbnd} for $r=0$. As was mentioned above, an extension to arbitrary $r$ is standard. 
\end{proof}
Let $\bar P^c:=\1-P^c$. Lemma \ref{lem:coer} and the self-adjointness of $K_0$  imply
\begin{corollary}[invertibility of $K_0$] \la{cor:invertK}
 For $R$ sufficiently large and and $r\ge 0$, the operator $\bar K_0:=\bar P^cK_0\bar P^c: \bar P^c\mH_{r+2}^c(\Om) \ra \bar P^c\mH_r^c(\Om)$ is invertible and its inverse, $\bar K_0^{-1}$, satisfies the estimate
\begin{equation}\label{Klowbnd-2}
 \|\bar K_0^{-1} w \|_{H^{r+2}}
  \leq c \| w \|_{H^{r}}.
\end{equation}
 This estimate, the definition of  $ G^c_\g,\ \g\in H_{per}^{2+r}(\Om)$, and the relations
 $\s L_0=  K_0\s$, $\s  P =  P^c\s$ and $\|\s w\|=\| w\|$ (see Section \ref{subsec:K}) imply that the operator $\bar L_0:=\bar P L_0\bar P$ is invertible and the inverse satisfies
 $\s \bar L_0^{-1}= \bar K_0^{-1}\s$ and \eqref{Linv}.
\end{corollary}

\bigskip

\DETAILS{\begin{lem}[coercivity] \la{lem:coer}
 For $R$ sufficiently large, $R\gg 1$, and for any $w \in ker (1-\Pz)$ and $s \ge 0$,
\begin{equation}\label{Llowbnd}
 \|\Lz w \|_{H^{s}}^2
  \geq c \| w \|_{H^{2+s}}^2.
\end{equation}
(For $n=1$ if $\kappa>\frac{1}{\sqrt{2}}$ and for any $n$ if $\kappa<\frac{1}{\sqrt{2}}$ we have a stronger statement: $c \|w\|_{H^1}^2 \leq \lan w, \Lz w \ran
  \leq \frac{1}{c} \| w \|_{H^1}^2$, which could be also adopted to a scale of Sobolev spaces.)
\end{lem}
\textbf{(Needs corrections!)} We prove the $n=1$ statement.  We omit the subindex $0$ in $L_0$ and write $L := \Lz$. The upper bound is
straightforward, so we prove only the lower one: 
\begin{equation}\label{coercbnd}
  \lan w, L w\ran \;\; \geq \;\; c \| w \|_{H^1}^2,
\end{equation}
for  $w$ 
orthogonal to approximate translational zero-modes, $G,\ \Tz_{k},\ k=1, 2$.

  As in Section~\ref{sec:zero}, we write
$L = L^{\nj} + V^{\nj}$, where recall $V^{\nj}$ satisfies the estimate
\[
  |V^{\nj}(x)| \lesssim   e^{-R}.
\]
Let $\{ \chi_j \}$ be a partition of unity associated to the
ball of the radius $R/2$ and its exterior, i.e. $\sum_{j=0}^{ 1} \chi_j^2 = 1$,
$\chi_0$ is supported in the ball of  the radius $3R/5$  and $\chi_1$
is supported outside the ball of  the radius $R/2$.
By the IMS formula~(\cite{cfks}),
\[
   L = \sum \chi_j L \chi_j
   -2\sum |\nabla \chi_j|^2.
\]
We can choose $\{ \chi_j \}$ such that
$|\nabla \chi_j| \lesssim R^{-1}$.

We extend the function $w$ to an $L^2-$function on $\R^2$ for which we keep the same notation. 
Since $\chi_1$ is supported
 outside the ball of  the radius $R/2$,
\[
  \lan\chi_1 w, L \chi_1 w\ran \geq
  c_1 \|\chi_1 w\|_{H^1}^2,
\]
for some $c_1 > 0$. Thus
\[
  \lan w, L w\ran  \geq 
   \lan\chi_0 w, L^{\nj}  \chi_0 w\ran
  + c_1 \|\chi_1 w\|_{H^1}^2
  - C( e^{-R} + R^{-2}) \|w\|_{H^1}^2.
\]
Now, since $w \in ker (1-\Pz)$, we have that $w \perp {T}_{k},\ k=1,2,\ G$ 
and therefore, due to \eqref{approxzeromodes}, we have, for the vortex translational and gauge
zero-eigenfunctions,  $ {T}^{(n)}_{k},\ k=1,2$,  $ G^{(n)}$ 
of $L^{\nj}$,
\[
  |\lan{T}_{k}, \chi_0 w \ran| \lesssim e^{-R},\  |\lan G, \chi_0 w \ran| \lesssim e^{-R}.
\]
So by the $n$-vortex stability result of~\cite{gs1} (for $n = 1$), we have
\[
  \lan\chi_0 w, L^{\nj} \chi_0 w\ran
  \geq c_2 \|\chi_0 w \|_{H^1}^2 - Ce^{-R} \|w\|_2^2,
\]
 and so, for $R$ sufficiently large,
\be
\la{eq:coer-2}
  \lan w, L w\ran \geq
  [c_3\sum\|\chi_j w \|_{H^1}^2 - C(e^{-R} + R^{-2})] \| w \|_{H^1}^2
  \geq c \| w \|_{H^1}^2.
\end{equation}
This proves \eqref{coercbnd}.} 
%
%

\subsection{Nonlinearity estimate}\label{sec:tay} 

\begin{lem}\la{lem:nl1}
For any $r>r'+1,\ r'\ge 0$ and $w  \in H^r $,
\[
\| N_{\vz}(w) \|_{H^{r'}} \leq c_r 
(\|w\|_{H^{r}}^2+\|w\|_{H^{r}}^3),\]
\be \la{est:N-Lip'}
  \| N_{\vz}(w')- N_{\vz}(w)  \|_{H^{r'}}\\
 \leq c_r(\|w\|_{H^{r}} + \|w\|_{H^{r}}^2+\|w'\|_{H^{r}} + \|w'\|_{H^{r}}^2)
  \|w' -w\|_{H^{r}}.
\end{equation}
\end{lem}
\begin{proof} 
We prove only the first estimate. The second one is proved similarly. Explicitly, $N_{\vz}(w)$ is given by
 \begin{equation} \label{Nv}
	\begin{cases}
	N_{\vz}(w)_\psi =(2i\alpha\cdot\nabla_{A^{(n)}} +i\div\alpha)\xi+|\alpha|^2 (\Psi^{(n)}+\xi)-\kappa^2[2\Re(\bar  \Psi^{(n)}\xi)\xi+|\xi|^2( \Psi^{(n)}+\xi)], \\
	N_{\vz}(w)_a = (2\Re(\overline{\Psi^{(n)}}\xi)+|\xi|^2)\alpha-\Im ( \bar{\xi} \nabla_{A^{(n)}} \xi).
	\end{cases}
\end{equation}
The most problematic term in $N_{\vz}(w)$ is of the
form $\xi \nabla \xi$, so we will just
bound this one (the rest are straightforward). 
Using Sobolev embedding theorems of the type $\| \xi\|_\infty
  \lesssim \| \xi\|_{H^{s}},$ for any $s>1 $, etc, and using the Leibnitz-type property of fractional derivatives (see \cite{stein, steinweiss}), we obtain, for $r>r'+s>1$,
\[
\begin{split}
  \| \xi \nabla \xi \|_{H^{r'}}
  &\lesssim \| \xi\|_{H^{s}} \|\nabla \xi \|_{H^{r'}}+\| \xi\|_{H^{r'+s}} \|\nabla \xi \|_{H^{0}} \\
  &\lesssim \| \xi\|_{H^{s}} \| \xi \|_{H^{r'+1}}+\| \xi\|_{H^{r'+s}} \| \xi \|_{H^{1}},
\end{split}
\]
which gives $\| \xi \nabla \xi \|_{H^{r'}} \lesssim  \| \xi \|_{H^{r}}^2.$
\DETAILS{\[ \begin{split}
  \| \xi \nabla \xi \|_{H^{-r}}
  &= \sup_{\| \eta \|_{H^r} = 1} |(\eta, \xi \nabla \xi)|
  \leq \sup \| \eta \xi \|_2 \| \nabla \xi \|_2 \\
  &\leq c \sup \|\eta\|_p \|\xi \|_q \| \xi \|_{H^1}
  \leq c \| \xi \|_{H^1}^2
\end{split}\]
where $1/p + 1/q = 1/2$ and $q$ is taken large enough
so that $H^r \subset L^p$.}
\end{proof} 

\section{Proof of Theorems \ref{thm:L} and \ref{thm:K}} \la{sec:proofthmsLK}

%
\DETAILS{\[
  L_{\psi,A} \left( \begin{array}{c} \xi \\ B \end{array} \right)
  = \left( \begin{array}{c}
  [-\Delta_A + \kappa^{2}(2|\psi|^2-1)]\xi +
  \kappa^{2}\psi^2\bar{\xi}
  + i[2\nabla_A \psi + \psi\nabla]\cdot B  \\
  \Im( [\bar{\nabla_A\psi}-\bar{\psi}\nabla_A] \xi)
  + (-\Delta + \nabla\nabla + |\psi|^2) \cdot B
  \end{array} \right).\]}
 Let $U^\cL \equiv (\Psi, A)$ be the $\cL-$periodic solution of~(\ref{gle}) found in Theorem \ref{thm:existAL}. (In this section we omit the superindex $\cL$ in $\Psi^\cL, A^\cL$.)
The proofs of  Theorems \ref{thm:L} and \ref{thm:K} are identical and we give the proof of  Theorem \ref{thm:L}. It follows from the two propositions given below. 
Define the shifted gauge  zero modes,   $ G_{j\g}(x)= \Gz^{(n)}_{\g}(x-j)$,
where
  $ \Gz^{(n)}_{\g}(x)$ are the  gauge  zero modes of the linearized operator, $L^{(n)} := F'( U^{(n)})$, 
  given in \eqref{ngmodes}.
\subsection{Zero and almost zero modes}

 \begin{prop}[approximate zero-modes of $L$] \la{prop:Lapprzeromodes}
With definitions given in Subsections \ref{subsec:res1} and \ref{subsec:spec} and under the additional condition that $\g\in H_{per}^2(\R^2)$ is exponentially localized, $|\g(x)| \lesssim e^{- cR}$ for some $c>0$, 
we have
\be \la{eq:approxL}
\|L \Tz_{jk}\| \lesssim e^{-R},\ \|L G_{j\g}\| \lesssim e^{-R}.
\end{equation}
\end{prop}
\begin{proof}
For each $j\in\cL$, we  write $L = L_j + V_{j}$, where $L_j$ is the shifted vortex linearized operator
\[
  L_j := 
  F'( U^{\nj}(\cdot-j))\equiv L^{(n)}|_{x\ra x-j},
\]
and $V_j$ is a multiplication operator defined by this relation.  Due to the explicit form \eqref{Lexpl} of $L$, given in Appendix \ref{app:LK}, and the estimates
\be  \la{uL'}
  U^\cL (x) = U^{(n)}(x-\alpha) + O_{H^1}(e^{-R}),\ \mbox{on}\ \Omega+\alpha,\ \forall\alpha\in \cL.
\end{equation}  on  the $\cL-$periodic solution $U^\cL \equiv (\Psi, A)$ of~(\ref{gle}), given in  Theorem \ref{thm:existAL}, $V_j$ satisfies
\be \la{Vjbnd}
  |V_{j}(x)| \lesssim e^{-\delta R},\ \qquad \mbox{if}\ \qquad |x-j|\le \delta R.
\end{equation} By the definition, $L_j$ has the zero modes,  which are shifted translation and gauge  zero modes, $T^{(n)}_{k}(x) $, and
$ \Gz^{(n)}_\g(x),\ \g\in H_{per}^2(\Om)$, of $L^{(n)} := F'( U^{(n)})$,
  for the $n-$vortex $U^{(n)}:=( \Psi^{(n)}, A^{(n)})$:
 \[T_{jk}(x)=T^{(n)}_{k}(x-j),\ k=1,2,\ \mbox{and}\  G_{j\g}(x)= \Gz^{(n)}_{\g}(x-j),\]
\be \la{Ljzeromodes}
  L_j\Tz_{jk} = 0,\ L_j G_{j\g}=0. \end{equation}
This, the estimates \eqref{vortdecay} and the condition that $\g$ is exponentially localized,  yield that
\be \la{TjkGjbnds}
  |\Tz_{jk}|,\ |G_{j\g}| \lesssim e^{-\delta R},\ \qquad \mbox{if}\ \qquad |x-j|\ge \delta R. \end{equation}
 Using these estimates and using \eqref{Vjbnd}, \eqref{Ljzeromodes} and the relations
$L = L_j + V_{j}$, we obtain the estimates \eqref{eq:approxL}  of Proposition~\ref{prop:Lapprzeromodes}.
\end{proof} 

\bigskip
\subsection{Coercivity away from the translation and gauge modes}

\begin{prop} \la{prop:Lcoerc}
Under conditions of Theorem \ref{thm:L}, there is $c >0$ s.t.
\be  \la{hess}
\lan \eta, L\eta\ran \ge c\| \eta \|_{H^1}^2,\   
\end{equation}
for any $\eta \perp \Span\{T_{jk},\ k=1,2,\ \forall j\in \cL,\ G_{\g},\ \g\in H^2(\R^2, \R)\}.$
\end{prop}
\begin{proof} Recall that the lattice $\cL$ is defined in such a way that vortices are located at the centers of its cells. Let $\cL'$ be a shifted lattice having vortices at its verices and let $\cL'':=\cL'\cup \{\infty\}$. 
Let $\{ \chi_j,\ j\in \cL'' \}$ be a partition of unity associated to the balls of radius $R/3$, centered at the points of the lattice $ \cL'$, i.e.
$\chi_j,\ j\in \cL',$ are supported in the balls, $B(j, R/3)$, of the radius $R/3$
about $ j\in\cL',$ $\chi_\infty$
is supported in $\R^2/\bigcup_{j\in\cL'}B(j, R/4)$, i.e. away from all the vortices,  and $\sum_{j\in\cL''} \chi_j^2 = 1$. We can choose $\{ \chi_j \}$ such that
$|\nabla \chi_j| \lesssim R^{-1}$.
By the IMS formula~(\cite{cfks}),
\be  \la{Lims}
   L = \sum \chi_j L \chi_j
   -2\sum |\nabla \chi_j|^2.
\end{equation}
As in the previous subsection, we  write
$L = L_j + V_{j}$,  for each $j\in\cL'$. 
By our choice of $\{ \chi_j,\ j\in \cL' \}$, we have that
$\| V_{j} |_{\supp \chi_j} \|_{\infty} \lesssim e^{-R}$ (see  \eqref{Vjbnd}), and so, for $j\in\cL'$,
\[
  \lan\chi_j \eta, L \chi_j \eta\ran \geq
  \lan\chi_j \eta, L_j \chi_j \eta\ran - Ce^{-R} \|\chi_j\eta\|^2.
\]
Let $ \g_j(x)= \g(x-j)$. Since $\eta \perp  G_{\g},\ \g\in H^2(\R^2, \R),$ we have $\lan G_{j\g}, \chi_j \eta \ran=\lan G_{j\g}- G_{\g_j}, \chi_j \eta \ran+\lan G_{\g_j}, (\chi_j-1) \eta \ran$. By \eqref{eq:close} and the exponential localization of $\g\in H^2_{per}(\R^2, \R),$ the first term on the r.h.s. is $\lesssim e^{-R}$. By exponential localization of $G_{j\g}$ the same is true for the second term as well. Hence we obtain $ |\lan G_{j\g}, \chi_j \eta \ran| \lesssim e^{-R}$. Next, since $\eta \perp T_{jk},\ k=1,2,$ $\forall j\in \cL'$, and  $\|(1-\chi_j) \Tz_{jk}\|_2 \lesssim e^{-R},$ 
we have $|\lan T_{jk}, \chi_j \eta \ran| \lesssim e^{-R}$. To sum up, for $j\in\cL$, and for all $\g\in H^2(\R^2, \R),$ exponentially localized, we have that
\[
  |\lan T_{jk}, \chi_j \eta \ran| \lesssim e^{-R}\|\eta\|,\  |\lan G_{j\g}, \chi_j \eta \ran| \lesssim e^{-R}\|\eta\|.
\]
So by the $n$-vortex stability result of~\cite{gs1}
(for all $n$ if $\kappa <\frac{1}{\sqrt{2}}$ and for $n = 1$ if $\kappa >\frac{1}{\sqrt{2}}$), we have, for $R$ sufficiently large and $\forall j\in \cL'$,
\[
  \lan\chi_j \eta, L_j \chi_j \eta\ran
  \geq c_1 \|\chi_j \eta \|_{H^1}^2. 
\]

Also, since $\chi_\infty$ is supported away from all the lattice sites, where the vortices are centered, we have that
\[
  \lan\chi_\infty \eta, L \chi_\infty \eta\ran \geq
  c_2 \|\chi_\infty \eta\|_{H^1}^2,
\]
for some $c_1 > 0$. The above estimates together with \eqref{Lims} and the fact that $\supp \nabla\chi_j$ for different $j$'s do not overlap and therefore $\sum_j |\nabla\chi_j|\lesssim R^{-2}$, give, for $R$ sufficiently large,
\be \la{eq:coer-3}
  \lan\eta, L \eta\ran \geq
  [c_3 - C R^{-2}] \| \eta \|_{H^1}^2
  \geq c \| \eta \|_{H^1}^2.
\end{equation}
 Hence we have shown \eqref{hess}.
\end{proof}

Propositions \ref{prop:Lapprzeromodes}  and \ref{prop:Lcoerc}  imply Theorem \ref{thm:L}. $\Box$

Theorem \ref{thm:K} is obtained by replacing, in the proof above, $L$ with $K$.

\begin{rem}\label{rem:unif-kappa} One can modify the proof of proposition \ref{prop:w} to make $R_0$ uniform in $\kappa - 1/\sqrt{2}$. To this end  one would have to `project out' also the $(\kappa = 1/\sqrt{2}) -$ zero modes (see \cite{gs1}).
\end{rem}
\appendix



\section{Critical magnetic fields }\label{sec:crit-mf}

In superconductivity there are several critical magnetic fields,  two of which (the first and the second  critical magnetic fields) are of special importance:

$h_{c1}$ is the field at which the first vortex enters the superconducting sample.
\medskip

$h_{c2}$ is the field at which a mixed state bifurcates from the normal one. 

\noindent
(The critical field $h_{c1}$ is defined as $h$ for which $G_Q(\Psi_s, A_s)=G_Q(\Psi^{(1)}, A^{(1)})$, for $Q=\R^2$). 
 For type I superconductors $h_{c1} > h_{c2}$ and for type II superconductors $h_{c1} < h_{c2}$. In the former case, the vortex states have relatively large energies, i.e. are metastable, and therefore are of little importance.

For type II superconductors, there are two important regimes to consider: 1) average magnetic fields per unit area, $b$, are less than but sufficiently close to  $h_{c2}$,
\be \la{regime1} 0< h_{c2}-b \ll h_{c2}
\end{equation} 
and 2) the external (applied) constant magnetic fields, $h$,  are  greater than but sufficiently close to  $h_{c1}$,
\be \la{regime2} 
 0< h-h_{c1} \ll h_{c1}.
 \end{equation}
The reason the first condition involves $b$, while the second $h$ is that the first condition comes from the Ginzburg-Landau equations (which do not involve $h$), while the second from the Ginzburg-Landau Gibbs free energy. 

%
\DETAILS{For an external field $h_a>h_{c1}$ and with the intervortex interaction neglected, we have $G/|Q|=m(E^{(1)}-\Phi^{(1)} H_a)$, where $m$ 
is the number of vortices per unit area. Thus the behaviour of $G$, as a function of $m$, changes at $H_a=H_{c1}
:=\frac{E^{(1)}}{\Phi^{(1)}}$. If $H_a<H_{c1}$, then the Gibbs energy is lowered by having no vortices, and if $H_a>H_{c1}$, then the Gibbs energy is lowered by having vortices.}  
One of the differences between the regimes \eqref{regime1} and \eqref{regime2} is that $|\Psi|^2$ is small in the first regime (the bifurcation problem) and large in the second one. 
If a superconductor fills in the entire $\R^2$, then in the second regime, the average magnetic field per unit area, $b\ra 0$, as $h \to h_{c1}$. 

\bigskip

\section{The 
operators $L$ and $K$} \label{app:LK} 
This appendix combines the construction of the complex $K$ extension of $L$ and statement of its fiber decomposition and its properties, due to \cite{gs1}, which is essential to our analysis, with the proof of self - adjointness of $K$ and the group theoretical elucidation of the fiber decomposition of the operator $K$. 

\subsection{Explicit form of $L$ and $K$} \label{app:explLK}
First, we write out explicitly the operators $L$ and $K$ introduced in Subsection \ref{subsec:spec} and Section \ref{subsec:K} and discuss a different way to treat the operator $L$. 
In this section we write operators $L$ and $K$ for any solution $U=(\Psi, A)$ of \eqref{gle}. 
The arguments below are presented on $\R^2$ but are also applicable on $\Om$.

Let $\mathcal{R}$ be the operation of taking the real part. The operator $L$ is given explicitly as (\cite{gs1})
\begin{equation} \label{Lexpl}
L=\left[ \begin{array}{cc} L_{11}  & L_{12}
\\L_{21} & L_{22}
\end{array} \right],
\end{equation}
with
\begin{equation}  \label{Lij} 
\begin{cases}
L_{11} =-\Delta_{A}+\kappa^{2}(2|\Psi|^2-1)+\kappa^{2}\Psi^2 \cC,\\ 
 L_{12}=i[ (\nabla_{A}\Psi)+\nabla_{A}\Psi]\cdot =i [2(\nabla_{A}\Psi) +\Psi\nabla]\cdot,\\ L_{21}=-\cR i[\overline{(\nabla_{A} \Psi})-\overline{\Psi} \overline{\nabla_{A}}], \\
 L_{22}=-\Delta+|\Psi|^2.
\end{cases}
\end{equation}
(Here $(\nabla_{A}\Psi)$ stands for the function resulting in application $\nabla_{A}$ to $\Psi$, while $\nabla_{A}\Psi$ stands for the product of operators $\nabla_{A}$ and multiplication by $\Psi$.) 
To prove symmetry of $L$, we have
\begin{align*}
&\Re \int \bar{\xi}(2i\omega\cdot\COVGRAD{A}\Psi + i\Psi\DIV\omega) =\int -2\omega\cdot\Im(\bar{\xi}\COVGRAD{A}\Psi) - \Im(\bar{\xi}\Psi)\DIV\omega \\
&=\int -2\omega\cdot\Im(\bar{\xi}\COVGRAD{A}\Psi)
+ \Im(\bar{\xi}\nabla\Psi - \bar{\Psi}\nabla\xi)\cdot\omega \\
&=\int -2\omega\cdot\Im(\bar{\xi}\COVGRAD{A}\Psi)
+ \Im(\bar{\xi}\nabla\Psi - i\bar{\xi}\Psi A + i\bar{\Psi}\xi A - \bar{\Psi}\nabla\xi)\cdot\omega \\
&= \int -\omega\cdot\Im(\bar{\xi}\COVGRAD{A}\Psi) + \bar{\Psi}\COVGRAD{A}\xi)
\end{align*}

 To extend the operator $L$ to a complex-linear operator $K$ we recall $\alpha = \left( \begin{array}{c}  \alpha_1 \\ \alpha_2  \end{array} \right)
  \leftrightarrow \alpha^c = \alpha_1 - i\alpha_2$, use 
 the complex notation
\begin{equation}\label{dcompl}
	\partial = \partial_{x_1} - i\partial_{x_2} ,\ \qquad \partial_{A^c} = \p - iA^c,
\end{equation}
and introduce the complex conjugate, $\bar{A}$, of an operator $A$ as the operator $\bar{A}:=\mathcal{C} A \mathcal{C}$, where $\mathcal{C}$ denotes complex conjugation.
Straightforward calculations show that
\begin{equation*}
	\DIV\alpha = \frac{1}{2}\partial \bar{\alpha}^c + \frac{1}{2}\bar{\partial} \alpha^c,
\end{equation*}
\begin{equation*}
	2i\alpha\cdot\COVGRAD{A}\Psi
		= -i(\partial_{A^c}^*\Psi)\alpha^c + i(\partial_{A^c}\Psi)\bar{\alpha}^c,
\end{equation*}
and
\begin{equation*}
	- \Im(\bar{\xi}\COVGRAD{A}\Psi)^c
		= \frac{i}{2}(\overline{\partial_{A^c}^*\Psi})\xi
				+ \frac{i}{2}(\partial_{A^c}\Psi)\bar{\xi}.
\end{equation*}
In what follows we  drop  the superscript $c$ in $A^c$. Using the above relations one shows that the complex-linear extension, $K$, of the operator $L$, is given explicitly as
\begin{equation}\label{K}
K = \left( \begin{array}{cccc}
-\COVLAP{A}  + \kappa^2(2|\Psi|^2 -1) & \kappa^2\Psi^2
& -i(\partial_{A}^*\Psi)+ \frac{i}{2}\Psi\bar{\partial}  & i(\partial_{A}\Psi)+ \frac{i}{2}\Psi\partial \\
\kappa^2\overline{\Psi^2}	& \overline{-\COVLAP{A}}   + \kappa^2(2|\Psi|^2 -1)
& -i(\overline{\partial_{A}\Psi}) - \frac{i}{2}\bar{\Psi}\bar{\partial} & \frac{i}{2}(\overline{\partial_{A}^*\Psi}) - \frac{i}{2}\bar{\Psi}\partial \\
\frac{i}{2}(\overline{\partial_{A}^*\Psi}) + \frac{i}{2}\bar{\Psi}\partial_A & \frac{i}{2}(\partial_{A}\Psi) + \frac{i}{2}\Psi\overline{\partial_A^*}  & -\Delta + |\Psi|^2 & 0 \\
-\frac{i}{2}(\overline{\partial_{A}\Psi})	- \frac{i}{2}\bar{\Psi}\partial_A^* & -\frac{i}{2}(\partial_{A}^*\Psi) - \frac{i}{2}\Psi\overline{\partial_A} & 0 & -\Delta + |\Psi|^2
		\end{array} \right).
\end{equation}
It is not hard to check that $K$ restricted to vectors on the r.h.s. of \eqref{embed} gives $L^c$.
\DETAILS{
\begin{equation}\label{Kdef}
	K \FourByOne{\xi}{\phi}{\alpha}{\w} = \FourByOne
{-\COVLAP{a_\e}\xi - \lambda_\e\xi + (2\kappa^2 + \frac{1}{2})|\psi_\e|^2\xi
+ (\kappa^2 - \frac{1}{2})\psi_\e^2\phi	+ 2i\alpha\cdot\COVGRAD{a_\e}\psi_\e}
{-\overline{\COVLAP{a_\e}}\phi - \lambda_\e\phi + (2\kappa^2 + \frac{1}{2})|\psi_\e|^2\phi
+ (\kappa^2 - \frac{1}{2})\psi_\e^2\xi - 2i\alpha\cdot\COVGRAD{a_\e}\psi_\e}
		{(-\Delta  + |\psi_\e|^2) \alpha +i(\phi\COVGRAD{a_\e}\psi_\e)-\xi\overline{\COVGRAD{a_\e}\psi_\e}}
{(-\Delta  + |\psi_\e|^2) \w - i(\xi\COVGRAD{a_\e}\psi_\e)-\phi\overline{\COVGRAD{a_\e}\psi_\e}}.
\end{equation}
%
It is not hard to check that $K$ restricted to vectors on the r.h.s. of \eqref{embed} gives $L$.
\begin{equation} \label{eq:comp}
  M\left( \begin{array}{cc} \xi\\ \eta\\ \alpha^c \\ \beta^c  \end{array} \right)=??,
\end{equation}
where we identified $L^2({\R}^2;{\R}^2)$ with
$L^2({\R}^2;{\C})$ through the correspondence
\begin{equation}
\label{complexif}
  \alpha = \left( \begin{array}{c}  \alpha_1 \\ \alpha_2  \end{array} \right)
  \leftrightarrow \alpha^c \equiv \alpha_1 - i\alpha_2.
\end{equation}}
%

We consider the linearized operator $L$, 
on a space of pairs $(\Psi, A)$, satisfying the gauge condition
\begin{equation} \label{gchoice}
  \Im(\overline{\Psi} \xi) - \nabla \cdot \alpha=0.
\end{equation}
We mention a convenient way to treat the condition \eqref{gchoice} 
 by passing to a modified real-linear  operator $L_\#$, defined by  the quadratic form  (\cite{gs1}) 
\[
  \lan w, L_\# w \ran =
  \lan w, L_0 w \ran + \int_{\R^2} (\Im(\overline{\Psi} \xi) - \nabla \cdot \alpha)^2,
    \]
where $w = (\xi, \alpha) \in L^2( \R^2, \C )\times L^2(
\Omega,  \R^2)$. Clearly, $L_\#$
agrees with $L$ on the subspace of $L^2( \R^2, \C )\times L^2(
\R^2,  \R^2)$ specified by the gauge condition~(\ref{gchoice}). This
modification has the important effect of shifting the essential
spectrum away from zero. 
A straightforward computation gives the following expression for
$L_\# $:
\[
  L_\# \left( \begin{array}{c} \xi \\ \alpha  \end{array} \right) =
  \left( \begin{array}{c}
  [-\Delta_A + \frac{\kappa^2}{2}(2|\Psi|^2-1) + \frac{1}{2}|\Psi|^2]\xi
  + \frac{1}{2}(\kappa^2-1)\Psi^2\bar{\xi} + 2i\nabla_A \Psi \cdot  \alpha  \\
  2\Im[\overline{\nabla_A \Psi}\xi] + [-\Delta + |\Psi|^2] \alpha
  \end{array} \right).
\]
%
%
\DETAILS{In order to pass from real-linear operators to complex-linear ones,
we complexify the space $ L^2( \Omega, \C )\times L^2( \Omega,  \R^2)$
via
\begin{equation}
\label{eq:comp}
  (\xi, \alpha) \mapsto
  (\xi, \bar{\xi}, \alpha^c, \bar{\alpha}^c),
\end{equation}
where we identified $L^2(\Omega;{\R}^2)$ with $L^2(\Omega;{\C})$
through the correspondence
\begin{equation}\label{complex-al}
  \alpha = \left( \begin{array}{c}  \alpha_1 \\ \alpha_2  \end{array} \right)
  \leftrightarrow \alpha^c \equiv \alpha_1 - i\alpha_2.
\end{equation}
This leads to the space $[L^2(\Omega;{\C})]^4$.}
%
%
The complex-linear extension, $K_\#$, of $L_\# $, defined on
$[L^2(\R^2;{\C})]^4$, is given by 
\begin{equation}\label{Ksharp}
  K_\# =
  \mbox{ diag } \{ -\Delta_A, -\overline{\Delta_A}, -\Delta, -\Delta \}
  + V,
\end{equation}
where $V$ is the matrix-multiplication operator given, using the notation \eqref{dcompl}, by
\[
V = \left( \begin{array}{cccc}
\frac{\kappa^2}{2}(2|\Psi|^2-1) + \frac{1}{2}|\Psi|^2  &
\frac{1}{2}(\kappa^2-1)\Psi^2  &  -i(\partial_A^* \Psi)  &  i(\partial_A \Psi)  \\
\frac{1}{2}(\kappa^2 - 1)\overline{\Psi}^2  &
\frac{\kappa^2}{2}(2|\Psi|^2-1) + \frac{1}{2}|\Psi|^2  &
-i(\overline{\partial_A \Psi})  &  i(\overline{\partial_A^* \Psi}) \\
i(\overline{\partial_A^* \Psi})  &  i(\partial_A \Psi)  & |\Psi|^2  &  0  \\
-i(\overline{\partial_A \Psi})  &  -i(\partial_A^* \Psi)  &  0  &  |\Psi|^2
 \end{array} \right).
\]
The components of $V$ are bounded, and it follows from
standard results 
that $K_\#$ is a self-adjoint operator on $[L^2(\R^2;{\C})]^4$, with domain 
$  D(K_\#) = [H_2(\R^2;{\C})]^4.$ 

\subsection{Self-adjointness of $K_0$} \label{app:sa-K}
Next, 
we sketch a proof of
\begin{thm}\label{thm:decomp}
The operator $K_0,\ \mbox{defined by the expression \eqref{K} on}\ L^{2}(\Om, \C^4)$ with the domain $\mH_{2}^c(\Om),$ is self-adjoint.\end{thm}
\begin{proof}
Due to representation of the \eqref{Ksharp} type for $K$ and standard arguments, the question of self-adjointness for $K$ reduces to the same question for $\COVLAP{a_0}$. To prove the latter we proceed as in \cite{rs2}, Theorem X.28. Namely, we use that, by construction and properties \eqref{vortatzero} of $a_n$, $a_0$ is $C^1$ and the fact that 
since $-\COVLAP{a_0} \ge 0$, it suffices to show that $(-\COVLAP{a_0} +1)^*\xi = 0$ implies $\xi =0$, which is equivalent to showing that $(-\COVLAP{a_0} +1)\xi = 0,\ \xi\in L^2,$ (in the weak sense) implies $\xi =0$. Now, we use Kato's inequality $\Delta |\xi| \ge \Re[(\sign \xi) \COVLAP{a_0}\xi]$, where $(\sign \xi)(x)=\bar\xi(x)/|\xi(x)|$ if $\xi(x)\ne 0$ and $(\sign \xi)(x)=0$ if $\xi(x)= 0$, (see e.g.  \cite{rs2}, Theorem X.33). By this inequality, $\Delta |\xi| \ge \Re[(\sign \xi) \COVLAP{a_0}\xi]=|\xi|\ge 0$. Let now $\w_\del\ge 0$ be an approximation of identity and $f_\del :=\w_\del * |\xi|$. Then by the above $\Delta f_\del :=\w_\del * \Delta |\xi|\ge 0$ and therefore $\lan f_\del, \Delta f_\del\ran  \ge 0$. On the other hand, by integration by parts,  $\lan f_\del, \Delta f_\del\ran  \le 0$. Therefore we have $\lan f_\del, \Delta f_\del\ran = 0$, which implies  $ f_\del = 0$. Since  $ f_\del \ra |\xi|$, as $\del\ra 0$, we conclude that $ |\xi| = 0$.    This completes the argument. (For more more general results on self-adjointness of Schr\"odinger type operators on Hermitian vector bundles see \cite{bms}.)
\end{proof}
\bigskip

\subsection{Lower bound on $K_0$ away from vortices} \label{app:lb-out}

\begin{lem}\label{lem:lb-out}
For $R$ sufficiently large, there is $c_1 > 0$ s. t.  for any $w$ satisfying \eqref{gchoice} and  supported  outside the ball of  the radius $R/2$, we have that
\begin{equation}\label{lb-out}
  \| K_0  w\| \geq   c_1 \| w\|_{H^2},
\end{equation}
\end{lem}
\begin{proof} In this prove we omit the subindex $0$ in $K_0$. First we prove that $\| K  w\| \geq   c_1 \| w\|$. By the Schwarz inequality it suffices to show that $\lan  w, K  w \ran \geq c_1 \| w\|_{L^2}^2$. To prove the latter inequality we use that for any $w$ satisfying \eqref{gchoice}, $K$ and $K_\#$  induce the same quadratic form, $\lan  w, K  w \ran =\lan  w, K  w \ran$.
 \textbf{Observe} that $\| K  w\| \geq   \lan  w, K_\#  w \ran$ and estimate the r.h.s. of the latter expression. To this end we use the explicit construction of $\psi_0$ and $a_0$  or the estimates \eqref{v-Un-est} which imply that  outside the ball of  the radius $R/2$ 
\[\||\psi_0|^2 -1\|_\infty,\ \|(\partial_{a_0}^*\psi_0)\|_\infty,\  \|(\partial_{a_0}\psi_0)\|_\infty \le Ce^{-R}.\]
and the explicit expression for $K_\#$ which is given by \eqref{Ksharp}, with $\Psi$ and $A$ replaced by $\psi_0$ and $a_0$, to obtain that  outside the ball of  the radius $R/2$,
\begin{equation}\label{K0}
K_\# = \left( \begin{array}{cccc}
-\COVLAP{a_0}  + \frac{1}{2}(\kappa^2 + 1) & \frac{1}{2}(\kappa^2 - 1)\psi_0^2 & 0 & 0 \\
\frac{1}{2}(\kappa^2 - 1)\bar{\psi_0}^2 & \overline{-\COVLAP{a_0}}  + \frac{1}{2}(\kappa^2 + 1) & 0 & 0 \\
0 & 0 & -\Delta + 1 & 0 \\
0 & 0 & 0 & -\Delta + 1
		\end{array} \right) + O(e^{-R}),
\end{equation}
and therefore $\lan  w, K_\#  w \ran \geq c \| w\|^2 .$
 As was argued above this gives, by the Schwarz inequality,
 \[
  \| K_\#  w \| \geq c \| w\|. 
\]
Next,  \eqref{Ksharp} implies that for come $C>0$, $\| K_\#  w \| \geq  \frac{1}{2}\|\Delta w\| - C\| w\|.$ Writing  $\| K_\#  w \| = \delta \| K_\#  w \|+(1-\delta) \| K_\#  w \|$ and applying the second inequality to the first term and the first inequality to the second one and choosing $\delta$ appropriately (say $\delta= \frac{2c}{1+2c+C}$), we arrive at \eqref{lb-out}.
\end{proof}

\subsection{Fibre decomposition of $K_\#$}\label{sec:block}

Now we consider the operator $K_\#$ for the vortex solution $U^{(n)}=(\Psi^{(n)}, A^{(n)})$. We denote the resulting operator by $K^{(n)}_\#$ and present the important decomposition of $K^{(n)}_\#$, which is due to the fact that vortices are gauge equivalent under the action of rotation, i.e.,
$$\Psi(R_\alpha x) = e^{in\alpha}\Psi(x),\ R_{-\alpha}A(R_\alpha x) = A(x),$$
where $R_\al$ is counterclockwise rotation in $\R^2$ through the angle $\al$.
This property induces the following symmetry property of $K^{(n)}_\#$.
Let $\rho_n : U(1) \rightarrow Aut([L^2({\R}^2; {\C})]^4)$
be the representation whose action is given by
\[
  \rho_n(e^{i\theta})
  (\xi, \chi, \al, \beta)(x) =
  (e^{in\theta}\xi, e^{-in\theta}\chi, e^{-i\theta}\al,
  e^{i\theta}\beta) (R_{-\theta}x).
\]
It is easily checked that the linearized operator $K_\#^{(n)}$
commutes with  $\rho_n(g)$ for any $g \in U(1)$.
It follows that $K_\#^{(n)}$ leaves invariant the eigenspaces of $d\rho_n(s)$ for any $s \in i{\R} = Lie(U(1))$. (The representation of $U(1)$ on each of these subspaces is multiple to an irreducible one.)
This results in (fiber) block decomposition of $K_\#^{(n)}$, which is described below.  In particular, the translational zero-modes
each lie within a single subspace of this decomposition.  In what follows we write functions on ${\R}^2$ in polar coordinates, so that
\begin{equation}
\label{eq:polar}
  \cH^c(\R^2): = [L^2({\R}^2; {\C})]^4 =  [L^2_{rad} \otimes
  L^2({\bf S}^1; {\C})]^4
\end{equation}
where $L^2_{rad} \equiv L^2({\R}^{+}, rdr)$. Let $\mathcal{C}$ be the operation of taking the complex conjugate.

\begin{thm}\label{thm:decomp}
\begin{enumerate}[(a)]
\item Let $\cH_m := [L^2_{rad} ]^4$ and define $U : \cH^c(\R^2) \to \cH$, where $ \cH = \bigoplus_{m \in {\bf Z}}\cH_m$, 
so that on smooth compactly supported $v$ it acts by the formula
	\begin{equation*}
		(U v)_m(r) = J_m^{-1} \int_0^{2\pi} \chi_m^{-1}(\theta) \rho_n(e^{i\theta}) v(x) d\theta.
	\end{equation*}
where $\chi_m(\theta)$ are characters of $U(1)$, i.e., all homomorphisms $U(1) \to U(1)$ (explicitly we have $\chi_m(\theta)=e^{im\theta}$) and
\[J_m:\cH_m \ra e^{i(m+n)\theta} L^2_{rad} \oplus   e^{i(m-n)\theta} L^2_{rad} \oplus -i e^{i(m-1)\theta} L^2_{rad} \oplus  i e^{i(m+1)\theta} L^2_{rad} \]
 acting in the obvious way.
Then $U$ extends uniquely to a unitary operator.
	
\item Under $U$ the linearized operator around the vortex, $K_\#^{(n)}$,
decomposes as
\begin{equation}\label{eq:flip}
  UK_\#^{(n)} U^{-1}   = \bigoplus_{m \in {\bf Z}} K_m^{(n)},
\end{equation}
where the operators $K_m^{(n)}$ act on $\cH_m$ as $J_m^{-1}K_\#^{(n)} J_m$.

\item The operators $K_m^{(n)}$ have the following properties:
\DETAILS{passing to a rotated version, $M_m^{(n)}$,
of the operator $K_m^{(n)}$,
\[
  M_m^{(n)} \equiv \left\{ \begin{array}{cc}
            R K_m^{(n)} R^T  &  m \geq 0  \\
            R' K_m^{(n)} (R')^T  &  m < 0
            \end{array} \right.
\]
where
\[
  R = \frac{1}{\sqrt{2}} \left( \begin{array}{cccc}
      1 & 1 & 0 & 0 \\
      -1 & 1 & 0 & 0 \\
      0 & 0 & 1 & 1 \\
      0 & 0 & 1 & -1  \end{array} \right),    \;\;\;\;\;\;\;\;\;\;
  R' = \frac{1}{\sqrt{2}} \left( \begin{array}{cccc}
       1 & 1 & 0 & 0 \\
       1 & -1 & 0 & 0 \\
       0 & 0 & 1 & 1 \\
       0 & 0 & 1 & -1  \end{array} \right),
\]
we have}
%
\begin{equation}\label{eq:flip}
  K_m^{(n)} = R K_{-m}^{(n)}R^T,\ \mbox{where}\  R= \TwoByTwo{Q}{0}{0}{Q}, Q = \TwoByTwo{0}{\cC}{\cC}{0}
\end{equation}
\begin{equation}\label{eq:contspec}
  \sigma_{ess}(K_m^{(n)}) = [\min(1,\lambda), \infty),
\end{equation}
\begin{equation}\label{eq:mon}
 \mbox{for}\ |n|=1\ \mbox{and}\  m \geq 2,\ \quad K_m^{(n)} - K_1^{(n)} \geq 0\ \quad \mbox{with no zero-eigenvalue,}
\end{equation}
\begin{equation}\label{eq:mon}
K_0^{(n)} \geq c > 0\ \quad \mbox{for all}\ \kappa,
\end{equation}
\begin{equation}\begin{split}\label{L1}
K_1^{(\pm1)} &\geq 0 \  \mbox{with non-degenerate zero-mode given by}\\ T &:= (f'-\frac{n(1-a)}{r} f, f' +\frac{n(1-a)}{r} f, 2 n\frac{a'}{r}, 0).
\end{split}\end{equation}
\end{enumerate}
\end{thm}
\begin{proof} We prove (a) and (b). The properties \eqref{eq:flip} - \eqref{L1} in (c) were proven in \cite{gs1} (the latter paper did not articulate the construction in (a) and (b) explicitely).

A straightforward calculation shows that for $\hat{v}_m =  \int_0^{2\pi} \chi_m^{-1}(\theta) \rho_n(e^{i\theta}) v(x) \frac{d\theta}{2\pi}$, $\rho_n(e^{i\theta})\hat{v}_m = \chi_m(\theta)\hat{v}_m$, from which it follows that $\hat{v}_m$ lies in the range of $J_m$. Therefore $U$ is well-defined. We now calculate that for smooth compactly supported $v$,
\begin{align*}
	&\sum_{m\in\Z} \int_0^\infty \left\|  \int_0^{2\pi} \chi_m^{-1}(\theta) \rho_n(e^{i\theta}) v(x) \frac{d\theta}{2\pi} \right\|^2 rdr \\
	&= \int_0^\infty  \int_0^{2\pi}  \int_0^{2\pi} \left( \sum_{m\in\Z} e^{im(\theta - \phi)} \right)  \overline{\rho_n(e^{i\theta}) v(x)} \rho_n(e^{i\phi}) v(x) rdr d\theta d\phi \\
	&= \int_0^\infty  \int_0^{2\pi}  |\rho_n(e^{i\theta}) v(x)|^2 rdr d\theta = \|v\|^2
\end{align*}
It then follows that $U$ extends to all of $\L^2(\R^2)^4$ with norm $\|U\| = 1$. To show that $U$ is in fact a unitary map, we consider the map $U^* : \mathscr{H} \to \L^2(\R^2)^4$ given by
\begin{align*}
	U^* g = \sum_{m\in\Z} J_m g_m.
\end{align*}
Similar calculations as above show that $U^*$ is indeed the adjoint of $U$ and also has norm $1$. This proves (a).

To prove (b), the essential fact is that $K_\#^{(n)}$ commutes with the $\rho_n$. We have for any $g = Uv \in \mathscr{H}$
\begin{align*}
	(UK_\#^{(n)}U^{-1}g)_m
	&= J_m^{-1} \int_0^{2\pi} \chi_m^{-1}(\theta) \rho_n(e^{i\theta}) K_\#^{(n)}  v(x) d\theta \\
	&= J_m^{-1}  K_\#^{(n)} \int_0^{2\pi} \chi_m^{-1}(\theta) \rho_n(e^{i\theta}) v(x) d\theta \\
	&= (J_m^{-1}  K_\#^{(n)} J_m) g_m.
\end{align*}
This then completes the proof of (b).
\end{proof}

Since, by ~(\ref{eq:contspec}) and \eqref{L1},
$K_1^{(\pm1)}|_{T^{\perp}} \geq \tilde{c} > 0$ and,
 by~(\ref{eq:mon}) and \eqref{L1}, $K_m^{(\pm1)} \geq c' > 0$
for $|m| \geq 2$, this theorem implies that
$K_\#^{(\pm1)} \geq c > 0$ on the subspace of $\cH^c(\R^2)$ orthogonal
to the translational zero-modes.





\begin{thebibliography}{9999}


\bibitem[Abr]{abr}
    A.~A.~Abrikosov,
    On the magnetic properties of superconductors of the second group,
    J. Exp. Theor. Phys. (USSR) \textbf{32} (1957), 1147--1182.

\bibitem[Ahlf]{Ahlfors}
    L.~V.~Ahlfors,
    \emph{Complex Analysis},
    McGraw-Hill, New York, 1979.

\bibitem[ABS1]{abs1}  S. Alama,  L. Bronsard and E. Sandier,
On the shape of interlayer vortices in the Lawrence--Doniach
model.  Trans. AMS, vol. 360 (2008), no. 1, pp. 1--34.
%
\bibitem[ABS2]{abs2} S. Alama, L. Bronsard, and E. Sandier, Periodic Minimizers of the  Anisotropic Ginzburg--Landau Model.  Calc. Var. Partial Differential Equations  vol. 36  (2009),  no. 3, 399--417.


\bibitem[ABS3]{abs3}  S. Alama,  L. Bronsard and E. Sandier, Minimizers of the Lawrence--Doniach functional with oblique magnetic fields, \emph{Comm. Math. Phys.} (to appear).


\bibitem[Al]{Al}
    Y.~Almog,
    On the bifurcation and stability of periodic solutions of the Ginzburg-Landau equations in the plane,
    \emph{SIAM J. Appl. Math.} 61 (2000), 149--171.

\bibitem[AS]{as}   H. Aydi, E. Sandier,
Vortex analysis of the periodic Ginzburg-Landau model, \emph{Ann. Inst. H. Poincar\'e Anal. Non Lin\'eaire}  26  (2009),  no. 4, 1223--1236.


\bibitem[BGT]{bgt}
    E.~Barany, M.~Golubitsky, and J.~Turski,
    Bifurcations with local gauge symmetries in the Ginzburg-Landau equations,
     \emph{Phys. D}  56  (1992), 36--56.

\bibitem[BMS]{bms} M. Braverman, O. Milatovic, and M. Shubin, Essential self-adjointness of Schr\"odinger-type operators on manifolds, Russian Math. Surveys 57:4 (2002), 641 -- 692.

\bibitem[CFKS]{cfks}
H. Cycon, R. Froese, W. Kirsch, B. Simon,
{\it Schr\"odinger Operators with Applications
to Quantum Mechanics and Global Geometry.}
Springer-Verlag (1987).

\bibitem[DFN]{dfn}
B. A. Dubrovin, A.T. Fomenko, S. P. Novikov, \emph{ Modern Geometry}, Springer.


\bibitem[Dut]{dut}
	M.~Dutour,
	Phase diagram for Abrikosov lattice,
	J. Math. Phys.  42  (2001), 4915--4926.

\bibitem[GS1]{gs1} S. Gustafson, I.M. Sigal, The stability of magnetic vortices, Comm. Math. Phys. {\bf 212} (2000) 257-275.


\bibitem[GS2]{gs2}
S. Gustafson, I.M. Sigal,
Effective dynamics of magnetic vortices,
Adv. in Math. {\bf 199} (2006) 448-498.




\bibitem[GST]{gst}
S.~J.~Gustafson, I.~M.~Sigal and T. Tzaneteas,
    Statics and dynamics of magnetic vortices and of Nielsen-Olesen (Nambu) strings,
    J. Math. Phys. 51, 015217 (2010).


\bibitem[Lash]{lash}
    G.~Lasher,
    Series solution of the Ginzburg-Landau equations for the Abrikosov mixed state,
	Phys. Rev. 140 (1965), A523--A528.

\bibitem[Odeh]{odeh}
    F.~Odeh,
    Existence and bifurcation theorems for the Ginzburg-Landau equations,
	J. Math. Phys. 8 (1967), 2351--2356.

\bibitem[OS]{os} Yu. Ovchinnikov and I.M. Sigal, Symmetry Breaking Solutions
to the Ginzburg-Landau Equation, {\it JETP} (2004),  1090--1108.

\bibitem[RSII]{rs2}
M.~Reed and B. Simon, \newblock {\em Methods of Modern Mathematical
Physics, II, Fourier Analysis, Self-Adjointness}. Academic Press, 1975.


\bibitem[SS]{ss}  E. Sandier and S. Serfaty, \textit{Vortices in the Magnetic Ginzburg--Landau Model.} {\sl Progress in Nonlinear Differential Equations and Their Applications, vol. 70.}  Birkh\"auser, Boston, 2007.

\bibitem[Stein]{stein} E.~M.~Stein \textit{Singular Integrals and Differentiability Properties of Functions}. Princeton, NJ: Princeton University Press, 1970.

\bibitem[SW]{steinweiss} E.~M.~Stein and G.~Weiss, \textit{Introduction to Fourier analysis on Euclidean spaces}. Princeton, NJ: Princeton University Press, 1971.




\bibitem[TS]{ts}
 T. Tzaneteas, I.M. Sigal,
 Abrikosov lattice solutions of the Ginzburg-Landau equations, {\it Spectral Theory and Geometric Analysis}, Contemporary Mathematics 535 (2011), 195-213, AMS,  arXiv.




\end{thebibliography}
\end{document}